\documentclass[english]{article}
\usepackage[T1]{fontenc}
\usepackage[utf8]{inputenc}
\usepackage{geometry}
\usepackage{color}
\usepackage{babel}
\usepackage{amsmath}
\usepackage{amsthm}
\usepackage{amssymb}
\usepackage{graphicx}
\usepackage{setspace}
\usepackage[unicode=true]
 {hyperref}

\makeatletter
\theoremstyle{remark}
\newtheorem{summary}{\protect\summaryname}
\theoremstyle{plain}
\newtheorem{assumption}{\protect\assumptionname}
\theoremstyle{plain}
\newtheorem{thm}{\protect\theoremname}
\theoremstyle{plain}
\newtheorem{question}{\protect\questionname}
\theoremstyle{definition}
\newtheorem{defn}{\protect\definitionname}
\theoremstyle{remark}
\newtheorem{rem}{\protect\remarkname}
\theoremstyle{plain}
\newtheorem{ax}{\protect\axiomname}
\theoremstyle{plain}
\newtheorem*{assumption*}{\protect\assumptionname}
\theoremstyle{plain}
\newtheorem*{question*}{\protect\questionname}
\theoremstyle{definition}
\newtheorem*{defn*}{\protect\definitionname}
\theoremstyle{plain}
\newtheorem*{ax*}{\protect\axiomname}

\makeatother

\providecommand{\assumptionname}{Assumption}
\providecommand{\axiomname}{Axiom}
\providecommand{\definitionname}{Definition}
\providecommand{\questionname}{Question}
\providecommand{\remarkname}{Remark}
\providecommand{\summaryname}{Summary}
\providecommand{\theoremname}{Theorem}

\begin{document}
\begin{doublespace}
\begin{center}
\textbf{\Large{}The Economics of Enlightenment: Time Value of Knowledge
and the Net Present Value (NPV) of Knowledge Machines, A Proposed
Approach Adapted from Finance}{\Large\par}
\par\end{center}

\begin{center}
\textbf{Ravi Kashyap }
\par\end{center}

\begin{center}
\textbf{SolBridge International School of Business / City University
of Hong Kong }
\par\end{center}

\begin{center}
\today
\par\end{center}

\begin{center}
Keywords: Knowledge Valuation; Time Value; Money; Finance; Net Present
Value; Economics; Enlightenment
\par\end{center}

\begin{center}
Journal of Economic Literature Codes: D81 Criteria for Decision-Making
under Risk and Uncertainty; G12 Asset Pricing; D8 Information, Knowledge,
and Uncertainty; D46 Value Theory
\par\end{center}

\begin{center}
Mathematics Subject Classification Codes: 91B06 Decision theory; 91B25
Asset pricing models; 68T30 Knowledge representation 
\par\end{center}
\end{doublespace}

\begin{center}
\textbf{\textcolor{blue}{\href{https://doi.org/10.1515/bejeap-2019-0044}{Edited Version: Kashyap, R. (2019). The Economics of Enlightenment: Time Value of Knowledge and the Net Present Value (NPV) of Knowledge Machines, A Proposed Approach Adapted from Finance. B.E. Journal of Economic Analysis and Policy, 20(2), 2019-0044. }}}
\par\end{center}

\begin{doublespace}
\begin{center}
\tableofcontents{}
\par\end{center}
\end{doublespace}
\begin{doublespace}

\section{Abstract}
\end{doublespace}

\begin{doublespace}
We formulate one methodology to put a value or price on knowledge
using well accepted techniques from finance. We provide justifications
for these finance principles based on the limitations of the physical
world we live in.\textcolor{black}{{} }We start with the intuition for
our method to value knowledge and then formalize this idea with a
series of axioms and models.\textcolor{black}{{} To the best of our
knowledge this is the first recorded attempt to put a numerical value
on knowledge.} The implications of this valuation exercise, which
places a high premium on any piece of knowledge, are to ensure that
participants in any knowledge system are better trained to notice
the knowledge available from any source. Just because someone does
not see a connection does not mean that there is no connection. We
need to try harder and be more open to acknowledging the smallest
piece of new knowledge that might have been brought to light by anyone
from anywhere about anything.\pagebreak{}
\end{doublespace}
\begin{doublespace}

\section{\label{sec:Knowledge-for-What}Knowledge for What Sake? }
\end{doublespace}

\begin{doublespace}
Despite the apparent lack of consensus as to what knowledge really
is and a possibly demonstrable absence of well accepted reasons for
maybe why we even need knowledge, most would agree that knowledge
is valuable\footnote{\begin{doublespace}
\label{fn:As-the-name}As the name of (section \ref{sec:Knowledge-for-What},
Knowledge for What Sake?) suggests, whatever knowledge might be if
it can be traded for whichever sake, the Japanese drink that we might
be able to ferret out, it might seem like a good exchange, or trade,
to some of us. Such a trade or barter or willingness to pay for a
commodity can help us establish the price or value for any good. The
main issue with such an approach, though it can be helpful, is that
it is an inaccurate estimate of the innate worth of the article being
transacted since it involves subjective decisions being made by the
participants including any dependencies on the particular situation
under which the transaction was performed (circumstances that might
have forced one of the participants to participate). Also, price estimation,
which involves predicting future benefits, is highly uncertain and
unreliable as evidenced by the enormous, growing and complex literature
on asset pricing (Cochrane 2009). Hence, our approach to estimating
the value of knowledge is much broader in scope than finding the price
for an exchange, which involves monetary or utilitarian connotations
(Foot-note \ref{fn:The-utility-maximization}). Price and value are
used synonymously in daily life and in most financial transactions.
But it should be clear that price and actual value are entirely distinct
for knowledge since someone truly looking for a certain piece of knowledge
will know that in some cases no amount of money can help obtain that
piece of knowledge.
\end{doublespace}
}\footnote{\begin{doublespace}
\label{fn:The-utility-maximization}The utility maximization problem
is the problem consumers face: \textquotedbl how should I spend my
money in order to maximize my utility?\textquotedbl{} \href{https://en.wikipedia.org/wiki/Utility_maximization_problem}{Utility Maximization, Wikipedia Link};
Within economics the concept of utility is used to model worth or
value, but its usage has evolved significantly over time. \href{https://en.wikipedia.org/wiki/Utility}{Utility, Wikipedia Link}.
\end{doublespace}
} . (Russell 1948) is perhaps an ideal place to start examining the
relation between individual experience and the general body of scientific
knowledge; see also: (Kvanvig 2003; Pritchard 2009; Pritchard, Millar
\& Haddock 2010; Crevoisier 2016)\footnote{\begin{doublespace}
\label{fn:Knowledge}Knowledge is a familiarity, awareness, or understanding
of someone or something, such as facts, information, descriptions,
or skills, which is acquired through experience or education by perceiving,
discovering, or learning. \href{https://en.wikipedia.org/wiki/Knowledge}{Knowledge, Wikipedia Link}
\end{doublespace}
}. (Boswell 1873; Toffler 1990)\footnote{\begin{doublespace}
\label{fn:Knowledge-Value}The idea that knowledge has value is ancient.
In the 1st century AD, Juvenal (55-130) stated “All wish to know but
none wish to pay the price\textquotedbl{} (Highet 1961; Courtney 2013).
In 1775, Samuel Johnson wrote: “All knowledge is of itself of some
value.” These quotes are from this Wikipedia link: \href{https://en.wikipedia.org/wiki/Knowledge_value}{Knowledge Value, Wikipedia Link},
which has many more quotes. The first time we checked the link was
around Oct-16-2017.
\end{doublespace}
} have a collection of quotes starting from the 1st century AD, stating
the ancient belief that knowledge has worth. That the Wikipedia link,
(Foot-note \ref{fn:Knowledge-Value}) about the value of knowledge,
has material only from as recently as 2000 years ago, shows how limited
our knowledge regarding the value history has placed on knowledge
is. (Dancy, Sosa \& Steup 2009; DeRose 2005; Figueroa 2016; Pritchard
2018; Hetherington 2018)\footnote{\begin{doublespace}
\label{fn:Epistemology}Epistemology is the study of the nature of
knowledge, justification, and the rationality of belief. It is the
branch of philosophy concerned with the theory of knowledge. \href{https://en.wikipedia.org/wiki/Epistemology}{Epistemology, Wikipedia Link}
\end{doublespace}
} are an excellent collection of articles on leading theories, thinkers,
ideas, distinctions and concepts in epistemology.
\end{doublespace}
\begin{summary}
\begin{doublespace}
\textbf{\textit{In later sections (and papers), we hope to pitch in
to this effort to refocus what knowledge stands for and how best to
spread it. }}\textbf{\textit{\textcolor{black}{Our contribution to
the knowns and unknowns about knowledge is to formulate one methodology
to put a value or price on knowledge}}}\textbf{\textit{ (Foot-note
\ref{fn:As-the-name})}}\textbf{\textit{\textcolor{black}{{} using well
accepted techniques from finance. }}}\textbf{\textit{We start with
the intuition for one method to value knowledge. We then formalize
this idea with a series of axioms and models.}}\textbf{\textit{\textcolor{black}{{}
To the best of our knowledge (limited as it is, which follows from
our present understanding of knowledge and also from the definition
of knowledge we use below, Definition \ref{def:Definition-Knowledge-is-a}),
this is the first recorded attempt to put a numerical value on knowledge
using well known valuation techniques. }}}\textbf{\textit{Our main
result (Theorem \ref{thm:Main-Result-}) provides a lower bound for
the value of knowledge. The implications of this valuation exercise,
which places a high premium on any piece of knowledge, are to ensure
that the participants in any knowledge system are better trained to
notice the knowledge available from any source. }}
\end{doublespace}
\end{summary}
\begin{doublespace}
In a related paper, (Kashyap 2018), we discuss a process to maximize
the efficiency of knowledge creation and show how our valuation method
signifies that the best way for journals to select submissions would
be randomly from a pool of papers meeting certain basic quality criteria.
We specifically show that the best decision we can make with regards
to the selection of articles by journals requires us to formulate
a cutoff point, or, a region of optimal performance and randomly select
from within that region of better results. The policy implication
(for all fields) is to randomly select papers, based on publication
limitations (journal space, reviewer load etc.), from an overall pool
of submissions that have a single shred of knowledge (or one unique
idea) and have the editors and reviewers coach the authors to ensure
a better final outcome. The results in this present paper set the
stage for the developments in (Kashyap 2018) where we compare existing
publication processes with an alternative approach based on randomization.

(Section \ref{sec:Questions-=000026-Answers,}) lays down the framework
required for our valuation exercise by providing the intuition, the
corresponding assumptions and definitions along with a review of the
relevant literature. (Section \ref{sec:Switching-on-Knowledge}) builds
upon this by making further assumptions and extends it to a formal
approach to include axioms, notation, terminology and the actual results.
We have tried to ensure that the bulk of the narrative within the
main body of the paper is mostly self-contained so that it can be
easily followed by a wider audience. But for those wishing to have
more details and a deeper context we have provided a rich set of foot-notes
that supplement the central arguments. We have provided a list of
all the definitions and assumptions including a dictionary of notation
and terminology in (appendices \ref{sec:Appendix-Definitions-Assumptions};
\ref{sec:Notation-and-Terminology}). 
\end{doublespace}
\begin{doublespace}

\section{\label{sec:Questions-=000026-Answers,}Questions \& Answers, Q\&A,
Definitions and Assumptions, D\&A, in our DNA}
\end{doublespace}

\begin{doublespace}
It would not be entirely incorrect to state that the majority of the
attempts at knowledge creation start with answering questions. In
present day society we seem to be focused on answering questions that
originate in different disciplines.
\end{doublespace}
\begin{assumption}
\begin{doublespace}
\label{assu:Hence,-as-a-field-categorization}As a first step, we
recognize that one possible categorization of different fields can
be done by the set of questions a particular field attempts to answer.
Since we are the creators of different disciplines, but we may or
may not be the creators of the world (based on our present state of
knowledge and understanding) in which these fields need to operate,
the answers to the questions posed by any domain can come from anywhere
or from phenomenon studied under a combination of many other disciplines.
\end{doublespace}
\end{assumption}
\begin{doublespace}
Hence, the answers to the questions posed under the realm of knowledge
creation can come from seemingly diverse subjects such as: physics,
biology, mathematics, chemistry, marketing, economics, finance and
so on. This quest for answers is bounded only by our imagination (Calaprice
2000)\footnote{\begin{doublespace}
\label{fn:Answering-questions-can}Answering questions can be also
viewed as solving problems that arise in any facet of life. This suggests
that we might be better off identifying ourselves with problems and
solutions, which tacitly confers upon us the title Problem Solvers,
instead of calling ourselves physicists, biologists, psychologists,
marketing experts, economists and so on.
\end{doublespace}
}. As we linger on the topic of Questions \& Answers, Q\&A. The field
that is most concerned with the valuation of assets is finance (for
lack of knowledge of a better word, or terminology, on behalf of the
authors, let us categorize knowledge under the umbrella of assets).
Hence, it should not come as a surprise that finance can provide a
very surprising answer to our main research question.
\end{doublespace}
\begin{question}
\begin{doublespace}
\textbf{\label{que:What-Knowledge-Value}What is the value of knowledge
in any field?}
\end{doublespace}
\end{question}
\begin{doublespace}
Any answer we wish to seek would depend on some Definitions and Assumptions,
D\&A. But if we change those D\&A we might get different Q\&A\footnote{\begin{doublespace}
\label{fn:Maybe,-DNA-hold}Maybe, we hold the lessons from the lives
of every ancestor we have ever had, even telling us that Q\&A and
D\&A might be in our very DNA, the biological one, which are always
changing (Alberts, Johnson, Lewis, Raff, Roberts \& Walter 2002; Alberts
2017). Evolution is constantly coding the information, compressing
it and passing forward, what is needed to survive better and to thrive,
building what is essential right into our genes. For information storage
in DNA and related applications see: (Church, Gao \& Kosuri 2012;
Lutz, Ouchi, Liu \& Sawamoto 2013; Kosuri \& Church 2014; Roy, Meszynska,
Laure, Charles, Verchin \& Lutz 2015). 

\label{DNA}Deoxyribonucleic acid (DNA) is a molecule composed of
two chains (made of nucleotides) that coil around each other to form
a double helix carrying the genetic instructions used in the growth,
development, functioning and reproduction of all known living organisms
and many viruses. \href{https://en.wikipedia.org/wiki/DNA}{DNA, Wikipedia Link}
\end{doublespace}
}.
\end{doublespace}
\begin{assumption}
\begin{doublespace}
\label{assu:Questions-=000026-Answers,}Questions \& Answers, Q\&A,
are important, but Definitions and Assumptions, D\&A, are even more
important since changing D\&A could require us to consider different
Q\&A.
\end{doublespace}
\end{assumption}
\begin{doublespace}
The implication of (assumption \ref{assu:Questions-=000026-Answers,})
is that all Q\&A we consider here are to be viewed in the context
of the D\&A provided here. While we endeavor to provide as realistic
and all encompassing a discussion as possible, a changing world (or
changes in our understanding of it) might render our Q\&A obsolete.
A related question that comes up is: what is finance? The answer is
that finance is a game where there are only three simple decisions
to be made: Buy, Sell or Hold; the complication are mainly to get
to these results (Kashyap 2015).
\end{doublespace}
\begin{doublespace}

\subsection{Related Literature}
\end{doublespace}

\begin{doublespace}
A recent attempt, (Martin 1996), in the context of the time periods
discussed in (section \ref{sec:Knowledge-for-What}), acknowledges
that measuring the value of knowledge has not progressed much beyond
an awareness that traditional accounting practices are misleading
and can lead to wrong business decisions. Right at the outset, we
distinguish between our knowledge valuation methodology and the valuation
of patents since all patents have some knowledge associated with them,
but not all knowledge might lead to a patent or an industrial application.
A patent is defined as any new or non-obvious invention capable of
industrial application (Pitkethly 1997). (Wu \& Tseng 2006) provide
a valuation technique based on real options. We also note that patents
are mostly granted based on demonstrated novelty and the associated
estimation of the future benefits of the corresponding idea or invention
tend to be extremely difficult. The question of novelty with respect
to knowledge and the distinction between old and new knowledge is
considered further in (section \ref{sec:Switching-on-Knowledge}).

While it is not straight forward to draw a distinction between basic
and applied research; many attempts at seeking knowledge are purely
for the sake of understanding a concept; the immediate utility of
such an undertaking is overlooked and perhaps not even recognized
when such efforts are undertaken (Shepard 1956; Nelson 1959; Reagan
1967; Pavitt 1991). (Foray \& Lundvall 1998) is a detailed discussion
of the economic impact of knowledge. (Barnett 1999) mentions that
even universities, which are solely meant to create and spread knowledge,
might have lost their way and how they need a new sense of purpose.
(Delanty 2001) is about the role of universities in the knowledge
society.

(Bozeman \& Rogers 2002) admit that determining the value of scientific
and technical knowledge poses a great many problems (the value of
knowledge shifts dramatically over time as new uses for the knowledge
emerge; a related problem is that market-based valuation of knowledge
is an inadequate index of certain types of scientific knowledge).
They present an alternative framework for the value of scientific
and technical knowledge, one based not on market pricing of information,
but instead on the intensity and range of uses of scientific knowledge.
Their churn model of scientific knowledge value emphasizes the distinctive
properties of scientific and technical knowledge and focuses on the
social context of its production. They consider the value of scientific
and technical knowledge in enhancing the activities of the set of
individuals who interact in the demand, production, technical evaluation,
and application of scientific and technical knowledge.

It is worth noting that goods that are not actively traded pose many
valuation challenges. There are many interesting techniques used to
determine the value of assets, especially non-financial ones. (Zhao
\& Zhou 2011) consider status indicators that determine wine prices
such as: individual wine tasting scores rated by critics and more
general status indicators like classified appellation affiliation,
extra designation on the label, or a winery's organizational status;
(Gustafson, Lybbert \& Sumner 2016) present an experimental approach
to measure consumer willingness to pay for wine attributes; (Costanigro,
McCluskey \& Mittelhammer 2007) provide empirical evidence that the
wine market is differentiated into multiple segments or wine classes
based on price ranges and find evidence that consumers value the same
wine attributes differently across categories.

(Ortiz, Stone \& Zissu 2017) develop a valuation model for a crowd
funding initiative to finance a pipeline to transport beer, which
takes into account the terms of the contract, the current and expected
future price of beer, and the life expectancy of the investor; (Masset
\& Weisskopf 2015) examine the performance, selectivity, and market-timing
abilities of wine fund managers; (Aznar \& Guijarro 2007) discuss
aesthetic variables in paintings; (Carmona 2015) is about jewelry
appraisal; (De Groot, Wilson \& Boumans 2002) present a conceptual
framework and typology for describing, classifying and valuing ecosystem
functions, goods and services; (Throsby 2003) is a discussion of the
non-market value of cultural goods - literature, visual arts, music,
theater, dance, etc.; (Lebreton, Jorge, Michel, Thirion \& Pessiglione
2009) isolate brain regions that may constitute a system that automatically
engages in valuating the various components of our environment so
as to influence our future choices.

As we go about applying finance principles to assess the value of
knowledge in all domains, we need to bear in mind that all valuations
are subjective since they are done by social beings and a hall mark
of the social sciences is the lack of objectivity. Here we assert
that objectivity is with respect to comparisons done by different
participants and that a comparison is a precursor to a decision. (Nagel
2012) clarifies the distinction between subjective and objective;
(Little 1991; Manicas 1991; Rosenberg 2018) discuss the philosophy
of the social sciences; (Kashyap 2017) points out that the social
sciences are objectively subjective; (Gerring 2011) is a detailed
account of social science methods used to provide explanations of
observed phenomena; (Harré 1985) explores the premise that knowledge
is a basis for moral good and relates various views about the nature
of science to different historical schools of philosophy; (Papineau
2002; O'hear 1993; Tweney, Doherty \& Mynatt 1981)\footnote{\begin{doublespace}
\label{enu:Philosophy-of-science}Philosophy of science is a sub-field
of philosophy concerned with the foundations, methods, and implications
of science. The central questions of this study concern what qualifies
as science, the reliability of scientific theories, and the ultimate
purpose of science. \href{https://en.wikipedia.org/wiki/Philosophy_of_science}{Philosophy of Science, Wikipedia Link}

\label{enu:The-scientific-method}The scientific method is an empirical
method of knowledge acquisition which has characterized the development
of science since at least the 17th century. \href{https://en.wikipedia.org/wiki/Scientific_method}{Scientific Method, Wikipedia Link}

\label{enu:philosophy-of-social-science}The philosophy of social
science is the study of the logic, methods, and foundations of social
sciences such as psychology, economics, and political science. \href{https://en.wikipedia.org/wiki/Philosophy_of_social_science}{Philosophy of Social Science, Wikipedia Link}
\end{doublespace}
} are detailed accounts of the philosophy and methodology of science.
\end{doublespace}
\begin{doublespace}

\subsection{Framework (D\&A) for Knowledge Valuation}
\end{doublespace}
\begin{assumption}
\begin{doublespace}
\label{assu:Despite-the-several}Despite the several advances in the
social sciences, we have yet to discover an objective measuring stick
for comparison, a so called, True Comparison Theory, which can be
an aid for arriving at objective decisions. Hence, despite all the
uncertainty in the social sciences, the one thing we can be almost
certain about is the subjectivity in all decision making.
\end{doublespace}
\end{assumption}
\begin{doublespace}
For our present purposes, the lack of such an objective measure means
that the difference in comparisons, as assessed by different participants,
can give rise to different valuations for the same element of knowledge
(or asset). We consider two extreme individuals and their perspectives,
which would influence their valuations. 
\end{doublespace}
\begin{defn}
\begin{doublespace}
\textbf{Type A} person, who has \textbf{All} the known knowledge in
the universe. So if any new knowledge becomes available he is desperate
to have it, since without this new knowledge he is incomplete.
\end{doublespace}
\end{defn}
\begin{doublespace}
We have to mention here that a type A person would place an extremely
high valuation on every element of knowledge he already has and any
new knowledge that might come up.
\end{doublespace}
\begin{defn}
\begin{doublespace}
\textbf{Type Z} person, who has no knowledge about anything in the
universe. So he cares nothing about any knowledge, wants nothing and
his valuation for all pieces of knowledge would be \textbf{Zero}. 
\end{doublespace}
\end{defn}
\begin{doublespace}
It is worth noting that there might be views expressed by people,
(or, people with beliefs in this very world that we live in who do
not seek anything not even knowledge), that once you stop looking
for things you will have everything. This would be a contradiction
to our definition, since in this case the type Z person is the one
who wants nothing and would put a valuation of zero on any new knowledge,
but since he wants nothing he would actually have everything and he
becomes the type A person. Alternately, the type A person, has all
the knowledge in the world which suggests that he wants nothing, or
his valuation of everything including knowledge is zero, making him
the type Z person. The scope of our discussion will be restricted
to someone who is between Type A and Z, so that he has a non-zero
valuation of every piece of knowledge.

The assumption made in finance regarding homogeneous expectations
(Levy \& Levy 1996; Chiarella \& He 2001), especially in the derivation
of many asset pricing models, investment analysis and portfolio management
principles (Bodie, Kane \& Marcus 2014), is stunningly sophisticated,
yet seemingly simplistic. Most people would argue that no two people
are alike, so this assumption does not seem validated (Valsiner 2007;
Buss 1985; Plomin \& Daniels 1987). Then again, rethinking this slightly
might lead to the following modified assumption.
\end{doublespace}
\begin{assumption}
\begin{doublespace}
The homogeneous expectations assumption in finance is perhaps a very
futuristic one where we are picking the best habits and characteristics
from our fellow beings (maybe not just humans?); and the environment
we live in and the external stimulus we receive tends to become more
similar (or we start to perceive it as more alike?), and at some point
in the future we might tend to have more in common with each other
fulfilling this great assumption, which seems more of a prophecy.
\end{doublespace}
\end{assumption}
\begin{doublespace}
There are many issues if we become too much like one another (Slatkin
1987; Frankham 1995; 1997)\footnote{\begin{doublespace}
\label{enu:Similar-Survival}If we become too similar, then mother
nature, or, evolution, will have less to work with. Since more differences
tell her, which traits are better for certain conditions; and many
possibilities create stronger survival potential. This is studied
under the label inbreeding. Not to mention, if we all look alike,
think alike and act alike the world would be a very mundane place. 

\label{enu:Inbreeding-is-the}Inbreeding is the production of offspring
from the mating or breeding of individuals or organisms that are closely
related genetically. \href{https://en.wikipedia.org/wiki/Inbreeding}{Inbreeding, Wikipedia Link}
\end{doublespace}
}. But with respect to finance we might evolve enough, so that one
day we might have the same expectations with respect to our monetary
concerns. This would also be the day when the Bid-Offer spread would
cease to matter, or, we would be indifferent to it making every coffee
shop, theater, street corner, pub, or everywhere … a venue for any
product (Kashyap 2015).
\end{doublespace}
\begin{assumption}
\begin{doublespace}
Using a related concept from economics regarding equilibriums (Dixon
1990; Varian 1992)\footnote{\begin{doublespace}
\label{enu:Equilibriums-should-perhaps}Equilibriums should perhaps
be more aptly named Quasi-Equilibriums since we never know what a
true equilibrium is; but perhaps any system fluctuates between multiple
equilibriums, somewhat like a see-saw: (Mantzicopoulos \& Patrick
2010; Stocker 1998).

\label{Seesaw}A seesaw (also known as a teeter-totter or teeter-board)
is a long, narrow board supported by a single pivot point, most commonly
located at the midpoint between both ends; as one end goes up, the
other goes down. \href{https://en.wikipedia.org/wiki/Seesaw}{Seesaw, Wikipedia Link}
\end{doublespace}
}, when we continue to evolve and evolve towards similarity, both the
\textbf{type} \textbf{A} and \textbf{type} \textbf{Z} kind of person
can be equilibriums, since they are (they become?) the same kind of
person with respect to their views on identifying the value of elements
around them.
\end{doublespace}
\end{assumption}
\begin{doublespace}
To better understand knowledge, let us first start with what is not
knowledge. Anything that we don't know is not knowledge\footnote{\begin{doublespace}
\label{enu:Confusion-Frustration}Many times what we don't know is
scary or can cause confusion or frustration. Confusion and Frustration,
though, scary and ugly to begin with, can be powerful motivators,
as long as, we don’t let them bother us. This is because: 
\end{doublespace}
\begin{itemize}
\begin{doublespace}
\item Confusion is the beginning of Understanding. 
\item Necessity, is the mother of all creation / innovation / invention,
but the often forgotten father, is Frustration, which is sometimes,
even more necessary, than necessity herself. Simply put, some amount
of frustration can be highly stimulating and lead to great possibilities.
\item What we learn from the story of, Beauty and the Beast (De Beaumont
1804), is that, we need to love the beasts to find beauty. Hence,
if we start to love these monsters (Confusion and Frustration), we
can unlock their awesomeness and find truly stunning solutions.
\item \label{enu:Beauty-Beast}Beauty and the Beast (French: La Belle et
la Bête) is a fairy tale written by French novelist Gabrielle-Suzanne
Barbot de Villeneuve and published in 1740 in The Young American and
Marine Tales (French: La Jeune Américaine et les contes marins). Her
lengthy version was abridged, rewritten, and published first by Jeanne-Marie
Leprince de Beaumont in 1756. \href{https://en.wikipedia.org/wiki/Beauty_and_the_Beast}{Beauty and the Beast, Wikipedia Link}
\end{doublespace}
\end{itemize}
}. We will further try to provide one definition for what knowledge
might be keeping in mind that as generic as we want to make any definition,
we need to be open to the possibility that the definition might need
to be altered depending on the specifics of the situation.
\end{doublespace}
\begin{defn}
\begin{doublespace}
\textbf{\textit{\label{def:Definition-Knowledge-is-a}Knowledge is
a connection between elements of this universe. The elements could
be many (more than two), two, or in some cases a link from one element
to the same element and all other combinations. This requires us to
clarify what is an element. We suggest that the element discussed
here is anything that belongs to this universe and any characteristic
of that element, as observable in this universe}}\footnote{\begin{doublespace}
\label{enu:We-emphasize-Universe}We emphasize the word universe since
our definition would need to be modified once we establish (there
is already speculation regarding this eventuality) the possibility
of other universes (Carr 2007; Weinberg 2007).

\label{Multiverse, Wikipedia Link}The multiverse is a hypothetical
group of multiple universes including the universe in which humans
live. \href{https://en.wikipedia.org/wiki/Multiverse}{Multiverse, Wikipedia Link}
\end{doublespace}
}\textbf{\textit{.}}
\end{doublespace}
\end{defn}
\begin{rem}
\begin{doublespace}
These connections can also be viewed as answers to appropriately posed
questions governed by suitable assumptions and definitions. Our present
endeavors in knowledge creation, or scientific research, can be understood
in this way. 
\end{doublespace}
\begin{rem}
\begin{doublespace}
Another possibility, which we consider in more detail in (section
\ref{subsec:Relativity-New-Knowledge}), is that knowledge, or the
connections between elements, exists with or without our observation
of those links. Many times our cognizance or understanding could be
incomplete or even incorrect and progressively gets better even though
the actual phenomenon itself has not changed.
\end{doublespace}
\begin{rem}
\begin{doublespace}
(Section \ref{subsec:The-Limits,-The}) considers the dimensions to
which we are (seem to be?) presently restricted to and the possibility
that any understanding of our universe is with respect to the limitations
imposed by the dimensions we are able to perceive. The possibility
of higher dimensions means the possibility of better comprehension
or it could simply be an altered view of perception from a different
number of dimensions (Kashyap 2019).
\end{doublespace}
\begin{rem}
\begin{doublespace}
It is important to emphasize here that the connection between elements
is not just at any particular point in time, but inter-temporally
and even across other higher dimensions. This also tells us that knowledge
from one time period is valuable for other time periods since they
could be linked. Knowledge across time periods can be useful should
there be a possibility of the same connection reoccurring, or, by
using the connection we know about to create a modified or new connection;
that is the value of knowledge is enhanced if we are able to use (or
reuse) existing connections discovered from another time period.
\end{doublespace}
\end{rem}
\end{rem}
\end{rem}
\end{rem}
\begin{defn}
\begin{doublespace}
\label{def:Knowledge-machines-are}Knowledge machines are elements
themselves that look to create, or discover, or record connections
between the various elements. They are people, research journals,
books, music, robots and everything else that fulfills the property
of being part of the efforts to add to the collective pool of knowledge.
We can also term them knowledge seekers.
\end{doublespace}
\end{defn}
\begin{doublespace}

\section{\label{sec:Switching-on-Knowledge}Switching on Knowledge Machines}
\end{doublespace}
\begin{doublespace}

\subsection{\label{subsec:Money-Machines-that}Money Machines that Run Secretly}
\end{doublespace}

\begin{doublespace}
In finance (Cochrane 2009), we talk about something called as a ``Money
Machine'' which will get turned off, as soon as people step in to
take advantage of it. This is also know as arbitrage (Shleifer \&
Vishny 1997; Brown \& Werner 1995)\footnote{\begin{doublespace}
\label{enu:Payoff-Arbitrage-LOOP}If payoff $A$ is always at least
as good as payoff $B$, and sometimes $A$ is better, then the price
of $A$ must be greater than the price of $B$. If two portfolios
(assets or goods) have the same payoffs (in every state of nature),
then they must have the same price. A few other subtleties can arise;
Arbitrage is possible when one of three conditions is met:
\end{doublespace}
\begin{itemize}
\begin{doublespace}
\item The same asset does not trade at the same price on all markets (\textquotedbl the
law of one price\textquotedbl ). 
\item Two assets with identical cash flows do not trade at the same price. 
\item An asset with a known price in the future does not trade today at
its future price discounted at the risk-free interest rate (or, the
asset has significant costs of storage; as such, for example, this
condition holds for grain but not for securities). 
\end{doublespace}
\end{itemize}
\begin{doublespace}
\label{Arbitrage, Wikipedia Link}In economics and finance, arbitrage
is the practice of taking advantage of a price difference between
two or more markets: striking a combination of matching deals that
capitalize upon the imbalance, the profit being the difference between
the market prices. When used by academics, an arbitrage is a (imagined,
hypothetical, thought experiment) transaction that involves no negative
cash flow at any probabilistic or temporal state and a positive cash
flow in at least one state; in simple terms, it is the possibility
of a risk-free profit after transaction costs. For example, an arbitrage
opportunity is present when there is the opportunity to instantaneously
buy something for a low price and sell it for a higher price. \href{https://en.wikipedia.org/wiki/Arbitrage}{Arbitrage, Wikipedia Link}
\end{doublespace}
} and it is possible when the law of one price is violated (Isard 1977;
Crouhy-Veyrac, Crouhy \& Melitz 1982; Protopapadakis \& Stoll 1983;
Goodwin, Grennes \& Wohlgenant 1990; Froot, Kim \& Rogoff 1995)\footnote{\begin{doublespace}
\label{Law of One Price, Wikipedia Link}The law of one price (LOOP)
states that in the absence of trade frictions (such as transport costs
and tariffs), and under conditions of free competition and price flexibility
(where no individual sellers or buyers have power to manipulate prices
and prices can freely adjust), identical goods sold in different locations
must sell for the same price when prices are expressed in a common
currency. This law is derived from the assumption of the inevitable
elimination of all arbitrage. \href{https://en.wikipedia.org/wiki/Law_of_one_price}{Law of One Price, Wikipedia Link}
\end{doublespace}
}. This is nothing but buying the asset that we think is cheaper than
what it should be and selling the asset that we think is more expensive
than what it should be. This is based on the price of the asset, or,
the market assessment of the asset in an applicable market, as of
today. When no price is available either due to the lack of a corresponding
market or participants, we can use the expectation of discounted future
cash flows.

There is a generally accepted concept in finance theory called the
time value of money (Ross, Westerfield \& Jaffe 2002; Bierman Jr \&
Smidt 2012; Delaney, Rich \& Rose 2016; Petters \& Dong 2016; Marty
2017; Olson \& Bailey 1981; Becker \& Mulligan 1997; Loewenstein \&
Prelec 1991; Figure \ref{fig:Time-Value-of})\footnote{\begin{doublespace}
\label{Time Value of Money}The time value of money is a theory that
suggests a greater benefit of receiving money now rather than later.
It is founded on time preference. \href{https://en.wikipedia.org/wiki/Time_value_of_money}{Time Value of Money, Wikipedia Link}

\label{Time Preference, Wikipedia Link}In economics, time preference
(or time discounting, delay discounting, temporal discounting) is
the current relative valuation placed on receiving a good at an earlier
date compared with receiving it at a later date. \href{https://en.wikipedia.org/wiki/Time_preference}{Time Preference, Wikipedia Link}
\end{doublespace}
}. This idea can be simply stated as the fact that most people (and
perhaps even animals with non-monetary rewards: Hayden 2016) would
rather have a certain sum of money now rather than the same sum of
money later. This intuitively makes sense to most people, since we
are not sure whether we will receive that certain sum in the future,
due to the main uncertainties that the future holds. Though, it can
be easily seen that for people that can travel through time, money
would be the same whether now or later. As unlikely or likely the
possibility of time travel might be, it is mainly being used here
to illustrate further the notion of why money has different values
at different periods of time (section \ref{subsec:The-Limits,-The}).
\end{doublespace}
\begin{doublespace}

\subsection{\label{subsec:Present-Future-Value-Money}Present Value and Future
Value of Money}
\end{doublespace}

\begin{doublespace}
For completeness, and to act as a reference point, we summarize the
formula for the time value of money. More mathematically advanced
treatments can use the concept of discount functions in both discrete,
or, continuous time with exponential or hyperbolic discounting, among
other possibilities, including the usage of differential equations;
see: (Thaler 1981; Frederick, Loewenstein \& O'donoghue 2002; Marzilli
Ericson, White, Laibson \& Cohen 2015; Cruz Rambaud \& Ventre 2017;
Figure \ref{fig:Discount-Functions:-Hyperbolic})\footnote{\begin{doublespace}
\label{Discount Function} A discount function is used in economic
models to describe the weights placed on rewards received at different
points in time. \href{https://en.wikipedia.org/wiki/Discount_function}{Discount Function, Wikipedia Link}
\end{doublespace}
}\footnote{\begin{doublespace}
\label{enu:Hyperbolic-Discounting}In economics, hyperbolic discounting
is a time-inconsistent model of delay discounting. Hyperbolic discounting
is mathematically described as 
\[
{\displaystyle g(D)={\frac{1}{1+kD}}\,}
\]
 where $g\left(D\right)$ is the discount factor that multiplies the
value of the reward, $D$ is the delay in the reward, and $k$ is
a parameter governing the degree of discounting. This is compared
with the formula for exponential discounting: ${\displaystyle f\left(D\right)=e^{-kD}\,}$
. \href{https://en.wikipedia.org/wiki/Hyperbolic_discounting}{Hyperbolic Discounting, Wikipedia Link}
\end{doublespace}
}. The essence of the below equations (eqns: \ref{eq:PF-FV}; \ref{eq:PV-FV-Cumulative};
\ref{eq:PV-FV-Discount-Functions}; \ref{eq:PV-FV-Discount-Growth}),
with regard to money are that money in the future decreases in value
when it is measured in the present, or brought into the present, since
the interest rate is usually non-negative, $i\geq0$; to be precise,
let us term this discounting\footnote{\begin{doublespace}
\label{enu:Componding-or-Discounting-Justification}Compounding or
Discounting arises due to two possibilities (which we can also view
as principles in this case) one of which is direct and the other is
indirect. 
\end{doublespace}
\begin{itemize}
\begin{doublespace}
\item When we compound we are accumulating interest, which is the direct
return we get that justifies this increasing value we assign to money
that we have today, which we could potentially invest. Likewise, when
we discount a sum of money we are likely to receive in the future,
we are taking away interest, which upholds our rule to decrease the
value of money. 
\item The indirect possibility is that our investment could become more
valuable not just due to the interest we receive but due to other
factors wherein our investments could intrinsically do well becoming
more valuable (other benefits such as good will, reduced taxes etc.
can occur due to making investments), which could act as indirect
returns to us when we have ownership in the investment. Similarly,
when we discount we are taking into account the risk that we might
never receive that sum of money in the future, which makes the value
of money lesser in an indirect way.
\item For simplicity, the general practice is to club together all these
possibilities into one number called the discount rate, which is related
to interest rates, inflation and risk among other things. This also
tells us that the rate for compounding and discounting could be different.
\end{doublespace}
\end{itemize}
}. Likewise, money from the present when it is to be valued in the
future, or taken into the future, increases in value; again for precision
sake, let us call it compounding. 

\begin{equation}
PV\ =\ \frac{FV}{\left(1+i\right)^{n}}\Longleftrightarrow FV\ =\ PV\left(1+i\right)^{n}\label{eq:PF-FV}
\end{equation}
Here, $PV$ (present value) is the value of money at $time=0$ or
at the present moment. $FV$ (future value) is the value of money
at $time=n$ or in the future. $n$ is the number of time periods
(not necessarily an integer). $i$ is the rate (interest rate) at
which money compounds each period. The cumulative present value of
multiple future cash flows can be calculated by summing the contributions
of $FV_{t}$, the value of future cash flows at time $t$, 
\begin{equation}
PV\ =\ \sum_{{t=1}}^{{n}}{\frac{FV_{{t}}}{(1+i)^{t}}}\label{eq:PV-FV-Cumulative}
\end{equation}
Expressed using discount functions, $f\left(\cdots\right)$, which
for money is less than one, that is $f\left(i,n\right)=1/\left(1+i\right)^{n}\leq1$,
we can write this as,
\begin{equation}
PV\ =\ FV\,f\left(i,n\right)\Longleftrightarrow FV\ =\ \frac{PV}{f\left(i,n\right)}\label{eq:PV-FV-Discount-Functions}
\end{equation}
We could include the growth rate of money, $g$, to depict any increase
or decrease in value not captured by interest rates. That is $g$
would be the growth rate of money over each time period. To get sensible
results, we usually require $g<i$ though $g$ can be positive or
negative. This would change the discount function, $f\left(\cdots\right)$,
as follows: $f\left(i,g,n\right)=\left(1+g\right)^{n}/\left(1+i\right)^{n}\leq1$
\begin{equation}
PV\ =\ FV\,f\left(i,g,n\right)\Longleftrightarrow FV\ =\ \frac{PV}{f\left(i,g,n\right)}\label{eq:PV-FV-Discount-Growth}
\end{equation}

\end{doublespace}
\begin{doublespace}

\subsection{\label{subsec:Once-On,-Never}Once On, Never Off}
\end{doublespace}

\begin{doublespace}
When an opportunity to make money is known and if many people become
aware of this opportunity, it usually disappears (Brealey, Myers,
Allen \& Mohanty 2012; Foot-notes \ref{Arbitrage, Wikipedia Link};
\ref{Law of One Price, Wikipedia Link}). \textbf{\textit{In contrast
to money (making), the most wonderful thing about knowledge (creation)
is that if more people know about it, the more switched on it will
be.}} Once we become acquainted with a connection, (either ourselves,
or due to the guidance of someone else), we generally see other links;
we either put a spin on the connection by relating it to other elements,
(new connections), or, we find other characteristics of the same connection,
which by our (definition \ref{def:Definition-Knowledge-is-a}) can
be viewed as new connections. Hence as more people become aware of
any piece of knowledge, they add more pieces of knowledge to it making
the overall body of knowledge grow with time.

That being said, the dissimilarities between money and knowledge do
not end there. \textbf{\textit{Money is not like knowledge in many
ways}}. For most of us money is quantifiable but knowledge is not.
It is hard to individuate or shred knowledge into bits and pieces
in a way that is convincing to most. Money requires collective agreement
and it has value due to this consensus; true knowledge does not require
acceptance, though recognition of it does. That knowledge has value
is usually not questioned, but to persuade everyone that knowledge
is like money is a tough ask; except perhaps for die-hard classical
economists, who can will themselves to believe that everything has
monetary value (Sandel 2012).

To overcome some of these objections, we start by looking at the current
mechanisms through which we capture, represent and store knowledge.
Knowledge is collected by doing research (Creswell 2008)\footnote{\begin{doublespace}
\label{enu:Research-has-been}Research has been defined in a number
of different ways, and while there are similarities, there does not
appear to be a single, all-encompassing definition that is embraced
by all who engage in it. \href{https://en.wikipedia.org/wiki/Research}{Research, Wikipedia Link}
\end{doublespace}
\begin{itemize}
\begin{doublespace}
\item One definition of research is used by the OECD, \textquotedbl Any
creative systematic activity undertaken in order to increase the stock
of knowledge, including knowledge of man, culture and society, and
the use of this knowledge to devise new applications.\textquotedbl{}
\item Another definition of research is given by John W. Creswell (Creswell
2002), who states that \textquotedbl research is a process of steps
used to collect and analyze information to increase our understanding
of a topic or issue\textquotedbl . It consists of three steps: pose
a question, collect data to answer the question, and present an answer
to the question.
\item The Merriam-Webster Online Dictionary defines research in more detail
as \textquotedbl studious inquiry or examination; especially : investigation
or experimentation aimed at the discovery and interpretation of facts,
revision of accepted theories or laws in the light of new facts, or
practical application of such new or revised theories or laws\textquotedbl{}
\item These definitions of research ignore the possibility that research
can be undertaken without the aim of devising new applications or
consciously collecting data and outside the defined boundaries of
our understanding of experiments and investigative processes; and
surely in ways we don't know how research can be done.
\item There is a ton of literature on how to make research economically
useful (Pavitt 1991); what are benefits and costs of research collaboration
at the individual, firm and international levels in the past, present
and future (Katz \& Martin 1997; Melin 2000; Beaver 2001; Miotti \&
Sachwald 2003); 
\end{doublespace}
\end{itemize}
} and the most fundamental units or tools of research can be thought
of as questions and answers. Knowledge is represented as concepts,
ideas or in a more rigorous manner as mathematical theorems. Knowledge
is deposited in journals, books, articles and the like. To provide
a monetary analogy, this would be like the different currencies, denominations
(dollars, cents, etc) and forms (credit card, cheques, cash, gold,
etc.) we use for money and store it in banks, safes at home, under
our mattress or wherever at times.

\textbf{\textit{This}} \textbf{\textit{simplification allows us to
count (introduce numbers) or use numerical methods to measure knowledge}}
that we have accumulated in familiar places \footnote{\begin{doublespace}
\label{Knowledge-Not-Seen}For brevity we ignore knowledge outside
the boundaries of our knowledge stores but available everywhere; for
example: growing on trees, hanging on statues, and so on, both literally
and figuratively. We justify this preclusion by stating that for most
of us there are no easy ways to snatch and store knowledge present
in kaleidoscopic forms; for those of us unable to spot knowledge handed
to us in papers (perhaps, sent to us for review), that do not match
our exact templates, page limits, artificial discipline boundaries
we have imposed and so no, it would be a hard ask to find knowledge
lying under a paddy field. It would be even harder to assess the impact
over time of this knowledge that has not yet been seen.
\end{doublespace}
}. The impact of this knowledge or its worth would be in how many new
connections it will spawn or how it will be connected to other elements
of knowledge as time passes. This is a conservative approach as it
would under-count and undervalue the knowledge we have gathered. We
are ignoring the links to knowledge outside our stockpiles since we
do not explicitly consider this knowledge outside our familiar stash.
We would like to highlight that we are not trying to put an exact
monetary value on knowledge. For example, if someone has come up with
a new theorem and published it in an article, our valuation is not
the exact amount of money to pay for their contribution; though our
techniques are a numerical approach at putting a value, they will
not provide a literal price for any knowledge transaction (our main
result, theorem \ref{thm:Main-Result-}, makes this clear).

\textbf{\textit{Another core problem is to attempt to value knowledge
only in terms of time elapsed}}. But if anything, as time passes the
breadth covered by any piece of knowledge, how well it integrates
into its related topics (or seemingly unrelated fields; assumption
\ref{assu:Hence,-as-a-field-categorization}) and how much support
it has, will increase as new connections will get added. With additional
breadth, value will increase. This can be measured by $g$, the growth
rate of knowledge over each time period. Even by ignoring breadth
we are putting a lowest possible estimate on the value of knowledge.
It is possible that we have not yet uncovered additional connections,
that increase the breadth of applicability of this knowledge. If we
have reasons to believe that certain knowledge gets stale over time,
(section \ref{subsec:Relativity-New-Knowledge}) considers the plausibility
of this alternative.
\end{doublespace}
\begin{doublespace}

\subsection{\label{subsec:Relativity-New-Knowledge}The Relativity of Generally
New, But Not So Specially New Theories of Knowledge}
\end{doublespace}

\begin{doublespace}
We consider the possibility that new knowledge could replace old knowledge
making it obsolete. (Shapere 1980; 1989) are fascinating accounts
of scientific change; (Fuchs 1993) considers a sociological theory
of scientific change that can address many problems including the
relationship between the natural and social sciences; (Laudan, Laudan
\& Donovan 1988), which is housed in (Donovan \& Laudan 2012), a larger
collection of empirical studies on scientific change, subject key
claims of some of the theories of scientific change to the kind of
empirical scrutiny that has been so characteristic of science itself;
(Laudan, ... \& Wykstra 1986) consider the role of historians and
their views on evaluating and improving theories of scientific change;
also relevant would be the many discussions on scientific realism:
(Bhaskar 1998; Psillos 2005; Bhaskar 2009; Smart 2014; Sankey 2016)\footnote{\begin{doublespace}
\label{enu:Scientific-realism-is}Scientific realism is the view that
the universe described by science is real regardless of how it may
be interpreted. \href{https://en.wikipedia.org/wiki/Scientific_realism}{Scientific Realism, Wikipedia Link}

\label{enu:A-paradigm-shift,}A paradigm shift, a concept identified
by the American physicist and philosopher Thomas Kuhn, is a fundamental
change in the basic concepts and experimental practices of a scientific
discipline. \href{https://en.wikipedia.org/wiki/Paradigm_shift}{Paradigm Shift, Wikipedia Link}
\end{doublespace}
}.

It is certainly tempting to think that newer more accurate theories
would render the older less accurate theories covering the same domain
as less useful. Newtonian mechanics (Halliday \& Resnick 1967)\footnote{\begin{doublespace}
\label{enu:Classical-or-Newtonian}Classical or Newtonian mechanics
provides extremely accurate results when studying large objects that
are not extremely massive and speeds not approaching the speed of
light. When the objects being examined have about the size of an atom's
diameter, it becomes necessary to introduce the other major sub-field
of mechanics: quantum mechanics (Shankar 2012). To describe velocities
that are not small compared to the speed of light, special relativity
is needed. In case that objects become extremely massive, General
relativity becomes applicable. \href{https://en.wikipedia.org/wiki/Classical_mechanics}{Newtonian Mechanics, Wikipedia Link}

\label{enu:Quantum-mechanics}Quantum mechanics (also known as quantum
physics, quantum theory, the wave mechanical model, or matrix mechanics),
including quantum field theory, is a fundamental theory in physics
which describes nature at the smallest scales of energy levels of
atoms and subatomic particles. \href{https://en.wikipedia.org/wiki/Quantum_mechanics}{Quantum Mechanics, Wikipedia Link}

\label{enu:General-relativity}General relativity (also known as the
general theory of relativity) is the geometric theory of gravitation
published by Albert Einstein in 1915 and the current description of
gravitation in modern physics. \href{https://en.wikipedia.org/wiki/General_relativity}{General Relativity, Wikipedia Link}
\end{doublespace}
} was incredibly useful and valuable but seems to have become less
so once Einstein’s theory of special relativity (Einstein 1956)\footnote{\begin{doublespace}
\label{enu:Special-relativity}Special relativity (also known as the
special theory of relativity) is the generally accepted and experimentally
well-confirmed physical theory regarding the relationship between
space and time. \href{https://en.wikipedia.org/wiki/Special_relativity}{Special Relativity, Wikipedia Link}
\end{doublespace}
}replaced it to cover the cases of objects moving at speeds closer
to the speed of light.
\end{doublespace}
\begin{itemize}
\begin{doublespace}
\item \textbf{\textit{First, we would like to distinguish between theory
and knowledge.}}
\end{doublespace}
\end{itemize}
\begin{doublespace}
We will continue with our mechanics example, but it will become clear
that this distinction applies more generally. A theory is a model\footnote{\begin{doublespace}
(Diallo, Padilla, Bozkurt \& Tolk 2013) make the case that theory
can be captured as a model.
\end{doublespace}
}, or a simplification of reality, based on various assumptions and
valid for a limited range of conditions (for discussions about the
use and abuse of models, see: Morgan \& Morrison 1999; Derman 2011;
Rescigno, Beck \& Thakur 1987)\footnote{\begin{doublespace}
\label{enu:A-theory-is}A theory is a contemplative and rational type
of abstract or generalizing thinking, or the results of such thinking.
\href{https://en.wikipedia.org/wiki/Theory}{Theory, Wikipedia Link}

\label{enu:A-conceptual-model}A conceptual model is a representation
of a system, made of the composition of concepts which are used to
help people know, understand, or simulate a subject the model represents.
\href{https://en.wikipedia.org/wiki/Conceptual_model}{Conceptual Model, Wikipedia Link}
\end{doublespace}
}. In Physics the models are usually mathematical that apply under
certain physical criteria \footnote{\begin{doublespace}
(Greca \& Moreira 2002) discuss the relationships among physical,
mathematical, and mental models in the process of constructing and
understanding physical theories.
\end{doublespace}
}. Newton's laws are mathematical models that are limited to non-relativistic
speeds (speeds much lower than the speed of light) and low gravitational
fields, and within those limits they provide exceptionally accurate
answers.

Saying that Newton was proved wrong by Einstein would be incorrect
in any sense. What relativity did was to expand the range of physical
conditions over which we can explain the movement of objects. Special
relativity extended the range to include high speeds, and general
relativity extended it again to include high gravitational fields.
At present, we do not seem to need theories that explain motion above
the speed of light\footnote{\begin{doublespace}
\label{enu:Tachyon}There is some awareness and much speculation about
objects that can travel faster than light. A tachyon or tachyonic
particle is a hypothetical particle that always moves faster than
light. Most physicists believe that faster-than-light particles cannot
exist because they are not consistent with the known laws of physics
(Bilaniuk, Deshpande \& Sudarshan 1962; Randall 2006). But surely,
beliefs change as unknown laws become known or as our knowledge increases.
\href{https://en.wikipedia.org/wiki/Tachyon}{Tachyon, Wikipedia Link}
\end{doublespace}
} and even general relativity is not applicable everywhere because
it fails at singularities like the center of black holes (Iyer \&
Bhawal 2013)\footnote{\begin{doublespace}
\label{enu:A-black-hole}A black hole is a region of space-time exhibiting
such strong gravitational effects that nothing, not even particles
and electromagnetic radiation such as light, can escape from inside
it. At the center of a black hole, as described by general relativity,
lies a gravitational singularity, a region where the space-time curvature
becomes infinite. \href{https://en.wikipedia.org/wiki/Black_hole}{Black Hole, Wikipedia Link}
\end{doublespace}
}. \textbf{\textit{To summarize, as new theories emerge the criteria
under which we can apply them to obtain answers, or knowledge, becomes
much larger.}}
\end{doublespace}
\begin{itemize}
\begin{doublespace}
\item \textbf{\textit{Second, we would like to emphasize that theory is
not knowledge}}\textbf{.}
\end{doublespace}
\end{itemize}
\begin{doublespace}
From (definition \ref{def:Definition-Knowledge-is-a}), Knowledge
exists independent of our awareness of it. Connections might exist
between objects without our knowledge of those relationships. Theory
only helps us to uncover or sense these connections. Then again, new
knowledge only becomes more useful by knowing the old knowledge or
connection, which shows how the new knowledge might be better and
in which situations the new theory would apply. If old knowledge is
forgotten or lost, new knowledge, might need to rediscover the old
connections, before its value becomes enhanced.
\end{doublespace}
\begin{itemize}
\begin{doublespace}
\item \textbf{\textit{Lastly, if any theory is shown to be completely incorrect,
it was never true knowledge in the first place. But old incorrect
theories are necessary since they will help us recognize what new
knowledge cannot be.}}
\end{doublespace}
\end{itemize}
\begin{doublespace}
When any theory becomes known to be wrong, we cannot discard it entirely.
We still need to know the old erroneous theory as we search for a
better understanding of the cosmos using the new theory that has replaced
the old theory as our tool. This is because, as we build upon the
new theory to overcome any limitations it might have and come up with
newer theories, we need to know what theories will not work. This
means, new knowledge becomes more valuable in yielding newer knowledge
when used together with old knowledge, even if it is now shown to
be wrong since it tells us what not to do. This can also be compared
to trial and error learning. We try and fail many times, but at last
when we succeed we have avoided the many ways in which we have failed
earlier to succeed this time\footnote{\begin{doublespace}
\label{enu:Taleb and Kahneman discuss Trial and Error / IQ Points}Such
unwanted outcomes creep up because we live in a world that requires
around 2000 IQ points, to consistently make correct decisions; but
the smartest of us has only a fraction of that (Ismail 2014). Hence,
we need to rise above the urge to ridicule the seemingly obvious blunders
of others, since without those marvelous mistakes the path ahead will
not become clearer for us. As Taleb explains, ``it is trial with
small errors that leads to progress''. That being said, if there
are big errors that might incapacitate the person trying the trial
from further trials; as long as someone else has observed the attempts
with huge errors, the rest of society benefits from it; assuming,
of course, that the big blow up has left a non-trivial portion of
society intact, or at-least not too shaken up. (Ismail 2014) mentions
the following quote from Taleb, “Knowledge gives you a little bit
of an edge, but tinkering (trial and error) is the equivalent of 1,000
IQ points. It is tinkering that allowed the industrial revolution''.
\href{https://www.youtube.com/watch?v=MMBclvY_EMA}{Nassim Taleb and Daniel Kahneman discuss Trial and Error / IQ Points, among other things, at the New York Public Library on Feb 5, 2013.}.
\end{doublespace}
}. An example of this would be about finding a way (or connection)
from one landmark in a city to another. It is easier to find newer
better paths if we know older worse paths that tell us which routes
not to take.

We will require newer and better theories as our understanding of
the world improves or as our knowledge increases. (Question \ref{que:Goal-Knowledge})
and associated remarks briefly consider the limiting case or the asymptotics
of increasing knowledge. The discussion above (sections \ref{subsec:Once-On,-Never};
\ref{subsec:Relativity-New-Knowledge}) suggests that the discount
function for knowledge should always be one or greater than one, stated
as the below axiom.
\end{doublespace}
\begin{ax}
\begin{doublespace}
\label{axm:Knowledge-never-decreases}Knowledge never decreases in
value. If it decreases in value, it was never knowledge to begin with.
This gives the following discount function (or weight function), $h\left(k,n\right)$,
for knowledge. $n$ is the number of time periods (not necessarily
an integer). $k$ is the rate (knowledge rate) at which knowledge
compounds each period.
\begin{equation}
h\left(k,n\right)=\left(1+k\right)^{n}\geq1\Rightarrow k\geq0\label{eq:Knowledge-Weight}
\end{equation}
\end{doublespace}
\end{ax}
\begin{rem}
\begin{doublespace}
Another way to look at change in knowledge is by considering the connections
any element of knowledge is likely to establish with other elements
of knowledge as time passes. The growth rate in the number of connections
to any element of knowledge can be expressed using the function in
(Eq. \ref{eq:Knowledge-Weight}). The discount function then tells
us how this piece of knowledge will link up with other elements of
knowledge. The discount function can be easily modified to include
the growth rate of knowledge, $g$, over each time period.
\end{doublespace}
\end{rem}
\begin{doublespace}
Instead of a discount function, since by the arguments in (sections
\ref{subsec:Once-On,-Never}; \ref{subsec:Relativity-New-Knowledge};
axiom \ref{axm:Knowledge-never-decreases}) knowledge never decreases
in value (irrespective of whether it moves forward or backward in
time), we will use the term weight. The weights to be used can be
based on the following assumptions:
\end{doublespace}
\begin{assumption}
\begin{doublespace}
\label{enu:Past-knowledge-is}Past knowledge is important the older
it is. That is knowledge from long ago is more valuable than knowledge
yesterday. 
\end{doublespace}
\end{assumption}
\begin{itemize}
\begin{doublespace}
\item This is because the knowledge we have today is built on the connections
we have created or discovered in the past. So the foundation for today
is the knowledge from long ago. The more knowledge we have from the
past the more consistent the knowledge we have today (that it is connected
to more elements making it important and crucial to new elements that
will get added); otherwise any alternate connections from the past
could render a large body of knowledge obsolete, but then again, only
by knowing alternate chains of connections, which though fallow, we
can create better connections. So knowledge is never wasted and it
never loses value.
\item A subtle point that we need to understand is that knowledge or the
connections between elements could stay the same over time, while
acknowledging that the awareness of the connections could be time
varying. Time is the dimension we do not have control over as per
our current understanding and methods of navigating the cosmos (section
\ref{subsec:The-Limits,-The} has more details). As we move through
time we uncover new connections. This implies that someone could discover
all the connections between various elements at some point in time,
but another person could come to know about these same connections
at another point in time. Their overall knowledge is the same, but
this awareness has happened at different points in time, which is
a constraint due to the physical limitations we presently have in
our Universe. The assumptions we make here is with respect to the
total recorded knowledge that is presently accumulated by all knowledge
seekers.
\item We can make a related assumption that knowledge becomes more important
at an increasing rate, going back from today (Rosser \& Lis 2016;
Courant, Robbins \& Stewart 1996; Wagenaar \& Sagaria 1975; Meadows,
Meadows, Randers \& Behrens 1972; Leike 2001; Foot-note \ref{Exponential Growth};
Figure \ref{fig:Exponential-Growth-Decay}). When any piece of knowledge
becomes most useful, its weight at that time can be captured as the
position of the center of the peak of a Gaussian function, or the
mean of a normal distribution and the weight at other times could
taper off similar to the normal density function: (Rao 1973; Cramér
2016; Foot-note \ref{Gaussian Function}; Figure \ref{fig:Gaussian-Weight-Functions}).
\end{doublespace}
\end{itemize}
\begin{assumption}
\begin{doublespace}
\label{enu:Future-knowledge-is}Future knowledge is important the
closer it is to us. That is knowledge that is far away from today
is less valuable than knowledge about tomorrow. 
\end{doublespace}
\end{assumption}
\begin{itemize}
\begin{doublespace}
\item If future knowledge cannot be immediately connected to the knowledge
we have now, or at any point in time, we fail to appreciate its value.
For example, if we took the knowledge of differential equations to
the past before we had knowledge of algebra or before zero was being
used, it would be very hard to see the importance of calculus. 
\item While this assumption seems counter intuitive, let us clarify with
an example: If someone tells us (correctly is another assumption which
assumes they are using the right theory capable of making this prediction)
that the world will end in 10 seconds versus 10 years. It would seem
that it does no good to know that the world will end in 10 seconds,
but knowing exactly when it will end in 10 years is incredibly valuable.
In this latter case, it would seem that future knowledge is more important
the ``further'' away from us it is, because there’s more that we
can do in preparation for the end of the world. But the knowledge
that the world will end in 10 years has to be accompanied by the knowledge
that it will not end in every smallest fraction of time before those
ten years. Without this knowledge, it is no good that we prepare for
ten years only to realize the world is gone well before those ten
years and will reestablish itself, perhaps in the same or in an altered
form. It is even likely that this could happen multiple times in those
ten years and hence all the knowledge within those ten years is essential
and to be more precise, more essential than the knowledge ten years
away.
\item We can make a related assumption that knowledge becomes less important
at an increasing rate going forward from today (Figure \ref{fig:Exponential-Growth-Decay})\footnote{\begin{doublespace}
\label{Exponential Growth}Exponential growth is exhibited when the
rate of change—the change per instant or unit of time—of the value
of a mathematical function is proportional to the function's current
value, resulting in its value at any time being an exponential function
of time, i.e., a function in which the time value is the exponent.
\href{https://en.wikipedia.org/wiki/Exponential_growth}{Exponential Growth, Wikipedia Link};
\href{http://mathworld.wolfram.com/ExponentialGrowth.html}{Exponential Growth, Mathworld Link}.
The exponential function ${\displaystyle x(t)=x(0)e^{kt}}\;,\;k>0$
, satisfies the linear differential equation: 
\begin{equation}
\!\,{\frac{dx}{dt}}=kx
\end{equation}
saying that the change per instant of time of $x$ at time $t$ is
proportional to the value of $x(t)$, and $x(t)$ has the initial
value $x(0)$. In the above differential equation, if $k<0$, then
the quantity experiences exponential decay.

\label{Exponential Decay}A quantity is subject to exponential decay
if it decreases at a rate proportional to its current value. \href{https://en.wikipedia.org/wiki/Exponential_decay}{Exponential Decay, Wikipedia Link};
\href{http://mathworld.wolfram.com/ExponentialDecay.html}{Exponential Decay, Mathworld Link}.
Symbolically, this process can be expressed by the following differential
equation, where $N$ is the quantity and $\lambda$ (lambda) is a
positive rate called the exponential decay constant: 
\begin{equation}
{\frac{dN}{dt}}=-\lambda N
\end{equation}
 The solution to this equation is: 
\begin{equation}
{\displaystyle N(t)=N_{0}e^{-\lambda t}}
\end{equation}
where $N(t)$ is the quantity at time $t$, and $N_{0}=N(0)$ is the
initial quantity, i.e. the quantity at time $t=0$.
\end{doublespace}
}; similar to what we discussed in the point above, (assumption \ref{enu:Past-knowledge-is}),
for past knowledge when any piece of knowledge becomes most useful
its weight at that time can be represented as the position of the
center of the peak of a Gaussian function, or the mean of a normal
distribution and the weight at other times could taper off similar
to the normal density function: (Figure \ref{fig:Gaussian-Weight-Functions})\footnote{\begin{doublespace}
\label{Gaussian Function} In mathematics, a Gaussian function, often
simply referred to as a Gaussian, is a function of the form: 
\begin{equation}
{\displaystyle f(x)=ae^{-{\frac{(x-b)^{2}}{2c^{2}}}}}
\end{equation}
 for arbitrary real constants $a,\;b$ and non zero $c$. It is named
after the mathematician Carl Friedrich Gauss. The graph of a Gaussian
is a characteristic symmetric \textquotedbl bell curve\textquotedbl{}
shape. \href{https://en.wikipedia.org/wiki/Gaussian_function}{Gaussian Function, Wikipedia Link};
\href{http://mathworld.wolfram.com/GaussianFunction.html}{Gaussian Function, Mathworld Link}

\label{Normal Distribution}In probability theory, the normal (or
Gaussian or Gauss or Laplace–Gauss) distribution is a very common
continuous probability distribution. The normal distribution is sometimes
informally called the bell curve. However, many other distributions
are bell-shaped (such as the Cauchy, Student's t, and logistic distributions),
\href{https://en.wikipedia.org/wiki/Normal_distribution}{Normal Distribution, Wikipedia Link}.
The probability density of the normal distribution is:
\begin{equation}
{\displaystyle f(x\mid\mu,\sigma^{2})={\frac{1}{\sqrt{2\pi\sigma^{2}}}}e^{-{\frac{(x-\mu)^{2}}{2\sigma^{2}}}}}
\end{equation}
where $\mu$ is the mean or expectation of the distribution (and also
its median and mode),$\sigma$ is the standard deviation, and $\sigma^{2}$
is the variance.
\end{doublespace}
}.
\item It is of course possible that knowledge from different times in the
far away future, that can shed light on connections we are working
with today, can become valuable; so there can be weights for future
knowledge that have multiple regions of significant value. This is
expressed by the use of multi-modal probability distributions: (Cramér
2016; Figure \ref{fig:Multi-Modal-Weight-Functions})\footnote{\begin{doublespace}
\label{Multimodal Distribution, Wikipedia Link}In statistics, a multi-modal
distribution is a continuous probability distribution with two or
more different modes. These appear as distinct peaks (local maxima)
in the probability density function, \href{https://en.wikipedia.org/wiki/Multimodal_distribution}{Multimodal Distribution, Wikipedia Link}.
Bi-modal distributions are an important sub-class, an important example
of a bi-modal distribution is the beta distribution. The probability
density function of the beta distribution, for $0\leq x\leq1$, and
shape parameters $\alpha,\beta>0$, is a power function of the variable
$x$ and of its reflection $(1-x)$ as follows: 
\begin{equation}
{\displaystyle {\begin{aligned}f(x;\alpha,\beta) & =\mathrm{constant}\cdot x^{\alpha-1}(1-x)^{\beta-1}\\[3pt]
 & ={\frac{1}{\mathrm{B}(\alpha,\beta)}}x^{\alpha-1}(1-x)^{\beta-1}
\end{aligned}
}}
\end{equation}
where the beta function, $\mathrm{B}$, is a normalization constant
to ensure that the total probability is $1$. In the above equations
$x$ is a realization, an observed value that actually occurred, of
a random process $X$. The beta function, also called the Euler integral
of the first kind, is a special function defined by 
\begin{equation}
{\displaystyle \mathrm{B}(x,y)=\int_{0}^{1}t^{x-1}(1-t)^{y-1}\,dt}\quad;x,y\text{ can be complex numbers with their real part greater than zero.}
\end{equation}
\end{doublespace}
}.
\end{doublespace}
\end{itemize}
\begin{assumption}
\begin{doublespace}
\label{enu:Certain-knowledge-could}Certain knowledge could be useful
at certain points in time, and then decays around a certain point
of high importance. 
\end{doublespace}
\end{assumption}
\begin{itemize}
\begin{doublespace}
\item We view this as impulse functions over an instantaneous period of
time or as impulse response functions whose effect persists over slightly
longer intervals of time around the time of the shock. That is when
certain knowledge becomes extremely valuable over a localized duration
of time and lesser in value (negligible in comparison to its value
when it is highest) as we move away that region of time (Nikodým 1966;
Hassani 2009; Hamilton 1994; Figure \ref{fig:Impulse-Functions};
\ref{fig:Impulse-Function-with-Linear})\footnote{\begin{doublespace}
\label{Dirac Delta or Impulse Function}In mathematics, the Dirac
delta function ($\delta$ function or impulse function) is a generalized
function or distribution introduced by the physicist Paul Dirac. It
is used to model the density of an idealized point mass or point charge
as a function equal to zero everywhere except for zero and whose integral
over the entire real line is equal to one, \href{https://en.wikipedia.org/wiki/Dirac_delta_function}{Dirac Delta or Impulse Function, Wikipedia Link};
\href{http://mathworld.wolfram.com/DeltaFunction.html}{Dirac Delta or Impulse Function, Mathworld Link}.
The Dirac delta can be loosely thought of as a function on the real
line which is zero everywhere except at the origin, where it is infinite,
\begin{equation}
\delta(x)={\begin{cases}
+\infty, & x=0\\
0, & x\neq0
\end{cases}}
\end{equation}
and which is also constrained to satisfy the identity 
\begin{equation}
\int_{-\infty}^{\infty}\delta(x)\,dx=1
\end{equation}

\end{doublespace}
\begin{enumerate}
\begin{doublespace}
\item This is merely a heuristic characterization. The Dirac delta is not
a function in the traditional sense as no function defined on the
real numbers has these properties. The Dirac delta function can be
rigorously defined either as a distribution or as a measure.
\item A closer look at the impulse function should tells us that we can
also consider suitable impulse response functions to capture the weight
we are assigning to how the value of any particular knowledge might
change over time.
\end{doublespace}
\end{enumerate}
\begin{doublespace}
\label{enu:Impulse-Response-Function}In signal processing, the impulse
response, or impulse response function, of a dynamic system is its
output when presented with a brief input signal, called an impulse.
Impulse response functions describe the reaction of endogenous macroeconomic
variables such as output, consumption, investment, and employment
at the time of the shock and over subsequent points in time. \href{https://en.wikipedia.org/wiki/Impulse_response}{Impulse Response, Wikipedia Link}

\label{enu:A-heuristic-technique}A heuristic technique (Ancient Greek:
\textquotedbl find\textquotedbl{} or \textquotedbl discover\textquotedbl ),
often called simply a heuristic, is any approach to problem solving,
learning, or discovery that employs a practical method, not guaranteed
to be optimal, perfect, logical, or rational, but instead sufficient
for reaching an immediate goal. \href{https://en.wikipedia.org/wiki/Heuristic}{Heuristic, Wikipedia Link}

\label{enu:Distributions-(or-generalized-functions)}Distributions
(or generalized functions) are objects that generalize the classical
notion of functions in mathematical analysis. Distributions make it
possible to differentiate functions whose derivatives do not exist
in the classical sense. \href{https://en.wikipedia.org/wiki/Distribution_(mathematics)}{Distributions or Generalized Functions, Wikipedia Link}

\label{enu:Measure-Mathematical-Analysis}In mathematical analysis,
a measure on a set is a systematic way to assign a number to each
suitable subset of that set, intuitively interpreted as its size.
In this sense, a measure is a generalization of the concepts of length,
area, and volume. Technically, a measure is a function that assigns
a non-negative real number or $+\infty$ to (certain) subsets of a
set $X$. \href{https://en.wikipedia.org/wiki/Measure_(mathematics)}{Measure, Wikipedia Link}
\end{doublespace}
}.
\end{doublespace}
\end{itemize}
\begin{doublespace}

\subsection{\label{subsec:Time-Value-of}Knowledge Valuation}
\end{doublespace}

\begin{doublespace}
With the above framework, which consists of various assumptions, definitions,
notation and terminology (all of which are summarized in appendices
\ref{sec:Appendix-Definitions-Assumptions}; \ref{sec:Notation-and-Terminology})
we present the answer to our primary research question (question \ref{que:What-Knowledge-Value})
as (theorem \ref{thm:Main-Result-}) below.
\end{doublespace}
\begin{thm}
\begin{doublespace}
\label{thm:Main-Result-} The value of all knowledge, or the value
of any knowledge at any point in time is growing. This means the past,
present and future value of every piece of knowledge is infinity.
\[
PRVK\ \triangleq\ FUVK\,h\left(k,n\right)=\infty\Longleftrightarrow FUVK\ \triangleq\ PRVK\,\bar{h}\left(\bar{k},n\right)=\infty
\]
\[
PRVK\ \triangleq\ PAVK\,h\left(k,n\right)=\infty\Longleftrightarrow PAVK\ \triangleq\ PRVK\,\bar{h}\left(\bar{k},n\right)=\infty
\]
\[
PAVK\ \triangleq\ FUVK\,h\left(k,n\right)=\infty\Longleftrightarrow FUVK\ \triangleq\ PAVK\,\bar{h}\left(\bar{k},n\right)=\infty
\]
Here, $PAVK,PRVK$ and $FUVK$ denote the past (PA), present (PR)
and future (FU) Value of any element of Knowledge. The notation $t\triangleq m$
means ``$t$ is equal by definition (under certain conditions) to
$m$ ''. $n$ is the number of time periods (not necessarily an integer).
$k$, $\bar{k}$ are the rates (knowledge rates) at which knowledge
compounds each period. $h\left(k,n\right)$ and $\bar{h}\left(\bar{k},n\right)$
are the weight functions and they satisfy the property given in (axiom
\ref{axm:Knowledge-never-decreases}) as below,
\[
h\left(k,n\right)=\left(1+k\right)^{n}\geq1\Rightarrow k\geq0
\]
\[
\bar{h}\left(\bar{k},n\right)=\left(1+\bar{k}\right)^{n}\geq1\Rightarrow\bar{k}\geq0
\]
\end{doublespace}
\end{thm}
\begin{proof}
\begin{doublespace}
Appendix \ref{sec:Proof-of-Theorem}.
\end{doublespace}
\end{proof}
\begin{rem}
\begin{doublespace}
\textit{This result provides a lower bound for the value of knowledge.
The implications of this valuation exercise, which places a high premium
on any piece of knowledge, are to ensure that participants in any
knowledge system are better trained to notice the knowledge available
from any source.}
\end{doublespace}
\end{rem}
\begin{doublespace}
We clarify further as to what $PAVK,\;PRVK$ and $FUVK$ mean. $PAVK$
or the Past Value of Knowledge, would refer to the value in the past
of any knowledge that we have today or we are likely to find in the
future. $PRVK$ would refer to the present value of any knowledge
that we had from the past or we are likely to find in the future.
$FUVK$ would refer to the future value of any knowledge that we have
today or we had from the past. 
\end{doublespace}
\begin{rem}
\begin{doublespace}
When a decision maker is confronted with two or more choices or options,
all of which provide infinite value, he is indifferent between the
choices and hence a random selection from the options available would
suffice.

The policy implications based on this lower bound for all efforts
regarding knowledge creation (which is the one of the main purposes
or perhaps even the sole purpose of all journals, researchers, etc.)
are discussed in (Kashyap 2018). When there are constraints (time
and other resources) a random selection from many qualified papers
with inputs from reviewers to improve the randomly selected papers
will be a more fruitful outcome for everyone involved and for society
as well. (Kashyap 2017) has a discussion of how this randomizing approach
is part of a bigger series of solution techniques to counter the uncertainty
we encounter in our lives.
\end{doublespace}
\end{rem}
\begin{doublespace}

\subsection{\label{subsec:The-Limits,-The}The Limits, The Physics and The Motivation
for Time Travel}
\end{doublespace}

\begin{doublespace}
The concept of time value of money (or the time value of anything)
arises because we (most of us?) do not have a way to travel across
time. Physics provides us with a theoretical basis for moving across
time (Lewis 1976; Deutsch \& Lockwood 1994; Woodward 1995; Nahin 2001;
Gott 2002; Arntzenius \& Maudlin 2002; Kaku 2009; Nahin 2017), which
is one of the dimensions of the physical world we live in. We are
four dimensional creatures: latitude, longitude, height and time are
our dimensions since we need to know these four co-ordinates to fully
specify the position of any object in our universe. This is perhaps
best made clear to lay audiences with regards to physics, such as
many of us, by the movie Interstellar (Thorne 2014). Also (Sagan 2006)
has a mesmerizing account of many physical aspects, including how
objects or beings can transform between worlds that are governed by
higher or lower dimensions and change their shapes or their physical
form when they enter a lower dimensional world though their shape
is unchanged in the higher dimensional world. As they move from the
higher dimensions into the lower dimensions, they would need to obey
the physical laws of the lower dimension, or take the physical shape
or be limited by the properties prescribed by the lower dimensional
world. This also implies that changing the number of dimensions based
on which we sense our environment could alter our understanding of
what is around us.

Time is our last dimension since it is the one which we cannot control
or move around in. But we can change the other three co-ordinates
(the first three co-ordinates are the co-ordinates of space, and we
can change where we are in space) and hence we have three degrees
of freedom. This points to a possibility that, beings from a higher
dimension can travel through time and enter our three dimensional
world; but they would need to be limited by the physical rules as
dictated by the dimensions of our world. This also suggests that one
way to enter our dimension from a higher dimensional universe (or
even a lower dimension) would be to be born in our world and likewise,
the way to leave our universe to go into a higher (or lower) dimension,
might be to die. And while it might be hard to take physical material
from one universe to another. Thoughts, or knowledge, might be mobile
across dimensions (Suddendorf \& Corballis 2007; Suddendorf \& Busby
2003; Suddendorf \& Busby 2005; Boyer 2008; Botzung, Denkova \& Manning
2008; Suddendorf, Addis \& Corballis 2011). And birth (or rebirth)
would be one way to travel across time. We could speculate further
that the reason for this cycle of births and deaths might be to collect
the most precious commodity that we can think of, which is knowledge,
since it is the only entity that gains in value as it moves across
time.

Now getting back to knowledge, knowledge now or later would still
be valuable (Theorem \ref{thm:Main-Result-}). If anything, for a
time traveler that travels around time collecting knowledge instead
of money, knowledge from different time periods would hold more value,
in comparison to money, which would have the same value at any point
in time for this time traveler\footnote{\begin{doublespace}
\label{enu:Inflation-Time-Traveler}The money a time traveler has
would not require an adjustment for simple things like inflation or
interest rates since the time traveler will travel to the point in
time where he could get the maximum value for his money. Say he wants
coffee and he has one dollar, he will go to the date, when he can
get the coffee that maximizes his utility (Jehle \& Reny 2011; Foot-note
\ref{fn:The-utility-maximization}) both in terms of quantity and
quality.
\end{doublespace}
}. 
\end{doublespace}
\begin{question}
\begin{doublespace}
\label{que:Goal-Knowledge}What is the ultimate goal of seeking knowledge?
\end{doublespace}
\end{question}
\begin{rem}
\begin{doublespace}
The motivation to be knowledge seekers is to obtain an understanding
of the connections between everything in this universe, through time
and other possible higher dimensions. This state of possessing complete
knowledge could be termed as ``Enlightenment'' (for a discussion
from recent history, see: Kant 1784; Foucault 1984; Norris 1994; Kapferer
2007)\footnote{\begin{doublespace}
\label{enu:Enlightenment-is-the}Enlightenment is the \textquotedbl full
comprehension of a situation\textquotedbl . \href{https://en.wikipedia.org/wiki/Enlightenment_(spiritual)}{Enlightenment, Wikipedia Link}

\label{enu:Implication-Knowledge-Seeker}The other implications of
being a knowledge seeker are:
\end{doublespace}
\begin{itemize}
\begin{doublespace}
\item Being aware of the possibility of attaining complete knowledge, is
that attempts at change might appear highly overrated and overstated.
We are not saying that change is unnecessary, just that trying to
change something without understanding it completely (that is, before
attaining enlightenment) might be a recipe for a disaster. If we understood
something and then wish to change it that might be warranted, but
trying to change something without knowing it fully is ridiculous
at best. 
\item Definitions \ref{def:Definition-Knowledge-is-a}, \ref{def:Knowledge-machines-are}
suggest that everything in this universe has knowledge associated
with it and / or it is contributing to the efforts at understanding
everything else. This makes everything important, worthy of respect
and infinitely valuable (Kashyap 2017; 2018 consider this in the case
of publications, job interviews and school admissions).
\item An unintended yet welcome consequence of our efforts to understand
everything in the universe could be that we might just end up understanding
one another better, perhaps, becoming more tolerant in the process;
making us wonder whether the the true purpose of all knowledge creation
might be to simply make us more tolerant. If someone is completely
ignorant and still highly tolerant then they have achieved the same
end goal as someone who goes about understanding everything and then
becomes tolerant. This would also be the case when our type A and
type Z person end up in identical states.
\end{doublespace}
\end{itemize}
}.
\end{doublespace}
\end{rem}
\begin{doublespace}

\section{Conclusion}
\end{doublespace}

\begin{doublespace}
We have formalized a methodology to estimate the value of knowledge
and established a lower bound using well established principles from
finance. In a related paper (Kashyap 2018), we show that our methodology
has significant policy implications for all fields of research and
also for all efforts aimed at the creation and dissemination of knowledge.
If one of the main purposes or perhaps even the sole purpose of all
journals, researchers, etc. is knowledge creation, our valuation suggests
that the best way for journals to select submissions would be randomly
from a pool of papers meeting certain basic quality criteria. Such
a random selection from many qualified papers, (after factoring in
constraints such as reviewer load, number of possible publications
etc.), with inputs from reviewers to improve the randomly selected
papers will be a more fruitful outcome for everyone involved and for
society as well. In such a scenario, the editors and reviewers are
not looking for ways to reject a paper instead they are coaching the
authors to ensure a better final end product.

Despite any rough edges in our discussion\footnote{\begin{doublespace}
\label{enu:To-remind-ourselves}To remind ourselves about one thing
on closing: \textquotedbl The source of most (all) human conflict
(and misunderstanding) is not because of what is said (written) and
heard (read), but is partly due to how something is said and mostly
because of the difference between what is said and heard and what
is meant and understood …\textquotedbl{} Part of the problem is also
due to the use of a very imprecise language like English for communication.
If the only paper written thus far in recorded history (to the best
of our knowledge) to attempt to find the value of knowledge using
numerical techniques, does not find a place in every journal it is
submitted to, then journals (and researchers) have lost touch with
their roots which is to find and disseminate knowledge (no matter
what form and shape it comes in).
\end{doublespace}
}, our only hope is that this paper represents the first step towards
a formal mechanism for putting the value on something, which we deem
valuable, but fail to recognize it in most places we see it (Foot-note
\ref{enu:Taleb and Kahneman discuss Trial and Error / IQ Points}).
Just because we do not see a connection does not mean that there is
no connection. Newton discovered gravity, which has existed since
time immemorial, only when an apple fell on his head. Though the truth
behind this myth can be debated, the metaphor is relevant for us,
since we sometimes only notice things that hit us: (McKie \& De Beer
1951; Fara 1999)\footnote{\begin{doublespace}
\label{Newton, Gravity, Apple}Squirreled away in the archives of
London’s Royal Society is a manuscript containing the truth about
the apple. It is the manuscript for what would become a biography
of Newton entitled Memoirs of Sir Isaac Newton’s Life written by William
Stukeley, an archaeologist and one of Newton’s first biographers,
and published in 1752 (Stukeley 1936). Newton told the apple story
to Stukeley, who relayed it as such: “After dinner, the weather being
warm, we went into the garden and drank tea, under the shade of some
apple trees… he told me, he was just in the same situation, as when
formerly, the notion of gravitation came into his mind. It was occasion’d
by the fall of an apple, as he sat in contemplative mood. Why should
that apple always descend perpendicularly to the ground, thought he
to himself…”. \href{https://www.newscientist.com/blogs/culturelab/2010/01/newtons-apple-the-real-story.html}{Newton Gravity Apple Linke}
\end{doublespace}
}. We need to try harder and be more open to acknowledging the smallest
piece of new knowledge that might have been brought to light\footnote{\begin{doublespace}
\label{Deepavali} Knowledge is synonymous to light in many cultures
across the world; for someone that is able to differentiate between
darkness and light, a single ray of light can be the greatest ally
in keeping hope alive; MacMillan 2008; Newman 2017; One of the most
popular festivals of Hinduism, Diwali the festival of lights, symbolizes
the spiritual \textquotedbl victory of light over darkness, good
over evil and knowledge over ignorance''. \href{https://en.wikipedia.org/wiki/Diwali}{Deepavali, Festival of Lights, Wikipedia Link}
\end{doublespace}
}by anyone from anywhere about anything and when there are constraints
(time and other resources) a random selection from many qualified
papers will lead to better utilization of resources and greater productivity
for everyone involved (Kashyap 2018). 

Success is a very relative term. In the extreme case, which we study
a bit about in finance as well, one person’s success (profit) could
be someone else’s failure (loss). That being said, to triumph in creating
a valuation for knowledge and almost everything else, it is important
to know where we are and start the journey towards where we want to
be. A consequence (perhaps unintended) of taking the first step on
a journey means that the percentage progress we have made, in terms
of the distance travelled, shoots up to infinity\footnote{Kashyap (2019b) provides an infinite progress benchmark to be successful
and has a detailed discussion.}. So once we start the trip it becomes manageable immediately. The
subjectivity in how we compare things (assumption \ref{assu:Despite-the-several})
means that the benchmark (assumption \ref{assu:Questions-=000026-Answers,}
can be used if benchmarks can be understood to fall under the category
of definitions and assumptions) for knowledge valuation might be constantly
changing; though our results show that, whatever the measuring stick
(or atleast a wide array of benchmarks), the value of every bit of
knowledge will remain infinite, which means that we need to keep on
learning till there is nothing left to learn or we no longer need
to learn anything.
\end{doublespace}
\begin{doublespace}

\section{\label{sec:Appendix-Definitions-Assumptions}Appendix: List of All
Definitions and Assumptions within the Paper}
\end{doublespace}
\begin{assumption*}
\begin{doublespace}
As a first step, we recognize that one possible categorization of
different fields can be done by the set of questions a particular
field attempts to answer. Since we are the creators of different disciplines,
but we may or may not be the creators of the world (based on our present
state of knowledge and understanding) in which these fields need to
operate, the answers to the questions posed by any domain can come
from anywhere or from phenomenon studied under a combination of many
other disciplines.
\end{doublespace}
\end{assumption*}
\begin{question*}
\begin{doublespace}
What is the value of knowledge in any field?
\end{doublespace}
\end{question*}
\begin{assumption*}
\begin{doublespace}
Questions \& Answers, Q\&A, are important, but Definitions and Assumptions,
D\&A, are even more important since changing D\&A could require us
to consider different Q\&A.
\end{doublespace}
\end{assumption*}
\begin{doublespace}
\rule[0.5ex]{1\columnwidth}{1pt}
\end{doublespace}
\begin{assumption*}
\begin{doublespace}
Despite the several advances in the social sciences, we have yet to
discover an objective measuring stick for comparison, a so called,
True Comparison Theory, which can be an aid for arriving at objective
decisions. Hence, despite all the uncertainty in the social sciences,
the one thing we can be almost certain about is the subjectivity in
all decision making.
\end{doublespace}
\end{assumption*}
\begin{defn*}
\begin{doublespace}
\textbf{Type A} person, who has \textbf{All} the known knowledge in
the universe. So if any new knowledge becomes available, he is desperate
to have it, since without this new knowledge he is incomplete.
\end{doublespace}
\end{defn*}
\begin{doublespace}
\rule[0.5ex]{1\columnwidth}{1pt}
\end{doublespace}
\begin{defn*}
\begin{doublespace}
\textbf{Type Z} person, who has no knowledge about anything in the
universe. So he cares nothing about any knowledge, wants nothing and
his valuation for all pieces of knowledge would be \textbf{Zero}. 
\end{doublespace}
\end{defn*}
\begin{assumption*}
\begin{doublespace}
The homogeneous expectations assumption in finance is perhaps a very
futuristic one where we are picking the best habits and characteristics
from our fellow beings (maybe not just humans?) and the environment
we live in and the external stimulus we receive tends to become more
similar (or we start to perceive it as more alike?), and at some point
in the future, we might tend to have more in common with each other,
fulfilling this great assumption, which seems more of a prophecy.
\end{doublespace}
\end{assumption*}
\begin{doublespace}
\rule[0.5ex]{1\columnwidth}{1pt}
\end{doublespace}
\begin{assumption*}
\begin{doublespace}
Using a related concept from economics regarding equilibriums (Dixon
1990; Varian 1992; Foot-note \ref{enu:Equilibriums-should-perhaps}),
when we continue to evolve and evolve towards similarity, both the
\textbf{type} \textbf{A} and \textbf{type} \textbf{Z} kind of person
can be equilibriums, since they are the same kind of person with respect
to their views on identifying the value of elements around them.
\end{doublespace}
\end{assumption*}
\begin{defn*}
\begin{doublespace}
\textbf{\textit{Knowledge is a connection between different elements
of this universe. The elements could be many (more than two), two
or in some cases, a link from one element to the same element and
all other combinations. This requires us to clarify what is an element.
We suggest that the element discussed here is anything that belongs
to this universe and any characteristic of that element, as observable
in this universe (Foot-note \ref{enu:We-emphasize-Universe}). }}
\end{doublespace}
\end{defn*}
\begin{doublespace}
\rule[0.5ex]{1\columnwidth}{1pt}
\end{doublespace}
\begin{defn*}
\begin{doublespace}
Knowledge machines are elements themselves that look to create, or
discover, or record, connections between the various elements. They
are people, research journals, books, music, robots and everything
else that fulfills the property of being part of the efforts to add
to the collective pool of knowledge. We can also term them knowledge
seekers.
\end{doublespace}
\end{defn*}
\begin{ax*}
\begin{doublespace}
Knowledge never decreases in value. If it decreases in value, it was
never knowledge to begin with. This gives the following discount function,
$h\left(k,n\right)$, for knowledge.
\begin{equation}
h\left(k,n\right)=\left(1+k\right)^{n}\geq1\Rightarrow k\geq0\label{eq:Knowledge-Weight-One}
\end{equation}
\end{doublespace}
\end{ax*}
\begin{assumption*}
\begin{doublespace}
Past knowledge is important the older it is. That is knowledge from
long ago is more valuable than knowledge yesterday. 
\end{doublespace}
\end{assumption*}
\begin{doublespace}
\rule[0.5ex]{1\columnwidth}{1pt}
\end{doublespace}
\begin{assumption*}
\begin{doublespace}
Future knowledge is important the closer it is to us. That is knowledge
that is far away from today is less valuable than knowledge about
tomorrow. 
\end{doublespace}
\end{assumption*}
\begin{doublespace}
\rule[0.5ex]{1\columnwidth}{1pt}
\end{doublespace}
\begin{assumption*}
\begin{doublespace}
Certain knowledge could be useful at certain points in time, and then
decays around a certain point of high importance. 
\end{doublespace}
\end{assumption*}
\begin{doublespace}

\section{\label{sec:Notation-and-Terminology}Appendix: Dictionary of Notation
and Terminology}
\end{doublespace}
\begin{itemize}
\begin{doublespace}
\item $PV$ (present value) is the value of money at $time=0$ or at the
present moment.
\item $FV$ (future value) is the value of money at $time=n$ or in the
future.
\item $n$ is the number of periods (not necessarily an integer).
\item $i,k$ are the rates (of money or knowledge, known as interest rate
or knowledge rate) at which the corresponding amount (of money or
knowledge respectively) compounds each period. 
\item $\bar{k}$ is also a knowledge rate.
\item $FV_{t}$, is the value of cash flows at time $t$.
\item $g$ is the growth rate of money or knowledge over each time period.
\item $f\left(\cdots\right)$, is the discount function, which for money
is less than one, that is $f\left(i,n\right)=1/\left(1+i\right)^{n}\leq1$
or $f\left(i,g,n\right)=\left(1+g\right)^{n}/\left(1+i\right)^{n}\leq1$
when the growth rate $g$ is included.
\item $h\left(k,n\right)$ or $\bar{h}\left(\bar{k},n\right)$ are the discount
(or weight) functions for knowledge which are greater than one, that
is, $h\left(k,n\right)=\left(1+k\right)^{n}\geq1\Rightarrow k\geq0$.
Strictly speaking since the weight functions for knowledge are greater
than one we need to write this as $h\left(k,n\right)=\left(1+k\right)^{n}>1\Rightarrow k>0$.
\item $PAVK,PRVK\text{ and }FUVK$ denote the past (PA), present (PR) and
future (FU) Value of any element of Knowledge.
\item $\triangleq$ is sometimes used in mathematics (and physics) for a
definition. The notation $t\triangleq m$ (often) means ``$t$ is
defined to be $m$'' or ``$t$ is equal by definition to $m$''
(often under certain conditions).
\end{doublespace}
\end{itemize}
\begin{doublespace}

\section{\label{sec:Proof-of-Theorem}Appendix: Proof of Theorem \ref{thm:Main-Result-}}
\end{doublespace}
\begin{proof}
\begin{doublespace}
We first try to establish the present value of a particular piece
of knowledge that will become known in the future.
\[
\text{Relationship between }PRVK\ \text{ and }\ FUVK\,
\]
We divide the time between the future and the present into many intervals.
The possibility of some knowledge becoming known in the future will
depend on the knowledge we have today and how we will add to it in
the intervals between now and later as we approach the future time
when the new knowledge will become known. The new knowledge rests
on the many connections it has with the knowledge we know today and
the knowledge that we will accumulate from now till the future point
in time. Even if the new knowledge does not depend on any of the existing
knowledge and its connections from now to the future, we would still
need all this knowledge from now to the future to be able to establish
that this new knowledge is independent of the other knowledge elements.
We can then consider that this new knowledge will have to be linked
(or still requires connections) to many pieces of knowledge in the
intervals between now till later or we need to be aware of all this
knowledge to establish that there is no relationship between the new
future knowledge and the knowledge in between. This can be viewed
as the change in the connections the new knowledge will have or the
growth in the new knowledge based on the connections it will have,
as given by the weight functions (assumptions \ref{enu:Past-knowledge-is};
\ref{enu:Future-knowledge-is}; \ref{enu:Certain-knowledge-could})
\[
PRVK\ \triangleq\ FUVK\,h\left(k,n\right)
\]
We then consider the limit as the number of intervals goes to infinity,
that is as $n\rightarrow\infty$. The present value is then given
by,
\[
PRVK\ \triangleq\underset{n\rightarrow\infty}{\lim}\ FUVK\,h\left(k,n\right)
\]
 Using the property of the weight function, 
\[
h\left(k,n\right)=\left(1+k\right)^{n}\geq1\Rightarrow k\geq0
\]
Also noting that since the knowledge across all those time intervals,
between the future and now, needs to be connected as discussed in
Assumption 7, the growth rate of knowledge is strictly greater than
zero or, $k>0$. By Axiom 1 every piece of knowledge has value and
hence the value of the piece of knowledge that becomes known in the
future when it becomes known is greater than zero, $FUVK>0$. We then
have,
\[
PRVK\ \triangleq\underset{n\rightarrow\infty}{\lim}\ FUVK\,\left(1+k\right)^{n}
\]
\[
PRVK\ \triangleq\ \infty
\]
Next, we consider the future value of any piece of knowledge that
we have today,
\[
\text{Relationship between }FUVK\ \text{ and }\ PRVK\,
\]
Using a similar method of dividing the time between the present and
future into many intervals, we consider the connections the piece
of knowledge we have today is likely to establish with other elements
of knowledge as time passes. As before, there will be direct connections
between the knowledge today and the knowledge we will accumulate based
on it; and to know that the knowledge today does not depend on other
knowledge will still require knowing about the other knowledge elements.
This can be viewed as the growth in this knowledge or change in this
knowledge based on the connections it will have as given by the weight
functions (assumptions \ref{enu:Past-knowledge-is}; \ref{enu:Future-knowledge-is};
\ref{enu:Certain-knowledge-could})
\[
FUVK\ \triangleq\ PRVK\,\bar{h}\left(\bar{k},n\right)
\]
We then consider the limit as the number of intervals goes to infinity,
that is as $n\rightarrow\infty$. The future value is then given by,
\[
FUVK\ \triangleq\underset{n\rightarrow\infty}{\lim}\ PRVK\,\bar{h}\left(\bar{k},n\right)
\]
Using the property of the weight function, 
\[
\bar{h}\left(\bar{k},n\right)=\left(1+\bar{k}\right)^{n}\geq1\Rightarrow\bar{k}\geq0
\]
Also noting that since the knowledge across all those time intervals,
between now and the future, needs to be connected as discussed in
Assumption 7, the growth rate of knowledge is strictly greater than
zero or, $k>0$. By Axiom 1 every piece of knowledge has value and
hence the value of the piece of knowledge that we have today is greater
than zero, $PRVK>0$. 
\[
FUVK\ \triangleq\underset{n\rightarrow\infty}{\lim}\ PRVK\,\left(1+k\right)^{n}
\]
\[
FUVK\ \triangleq\ \infty
\]
The relationships between the present value and past value; the past
value and present value; the past value and future value; and future
value and past value can be established using similar arguments as
above. To summarize, the below relationships can be shown to give
the results in the statement of (Theorem \ref{thm:Main-Result-}).
\[
\text{Relationship between }PRVK\ \text{ and }\ PAVK\,
\]
\[
\text{Relationship between }PAVK\ \text{ and }\ PRVK\,
\]
\[
\text{Relationship between }PAVK\ \text{ and }\ FUVK\,
\]
\[
\text{Relationship between }FUVK\ \text{ and }\ PAVK\,
\]
\end{doublespace}
\end{proof}
\begin{doublespace}

\section{\label{sec:References}References}
\end{doublespace}
\begin{enumerate}
\begin{doublespace}
\item Alberts, B. (2017). Molecular biology of the cell. Garland science.
\item Alberts, B., Johnson, A., Lewis, J., Raff, M., Roberts, K., \& Walter,
P. (2002). Molecular Biology of the Cell, Garland Science, New York.
\item Arntzenius, F., \& Maudlin, T. (2002). Time travel and modern physics.
Royal Institute of Philosophy Supplements, 50, 169-200.
\item Aznar, J., \& Guijarro, F. (2007). Modelling aesthetic variables in
the valuation of paintings: an interval goal programming approach.
Journal of the Operational Research Society, 58(7), 957-963.
\item Barnett, R. (1999). Realizing the university. McGraw-Hill Education
(UK).
\item Beaver, D. D. (2001). Reflections on scientific collaboration (and
its study): past, present, and future. Scientometrics, 52(3), 365-377.
\item Bhaskar, R. (1998). Philosophy and scientific realism. Critical realism:
Essential readings, 16-47.
\item Bhaskar, R. (2009). Scientific realism and human emancipation. Routledge.
\item Bierman Jr, H., \& Smidt, S. (2012). The time value of money. In The
Capital Budgeting Decision, Ninth Edition (pp. 29-59). Routledge.
\item Bilaniuk, O. M. P., Deshpande, V. K., \& Sudarshan, E. G. (1962).
“Meta” relativity. American Journal of Physics, 30(10), 718-723.
\item Bodie, Z., Kane, A., \& Marcus, A. (2014). Investments (10th global
ed.). Berkshire: McGraw-Hill Education.
\item Boswell, J. (1873). The Life of Samuel Johnson. William P. Nimmo.
\item Botzung, A., Denkova, E., \& Manning, L. (2008). Experiencing past
and future personal events: Functional neuroimaging evidence on the
neural bases of mental time travel. Brain and cognition, 66(2), 202-212.
\item Boyer, P. (2008). Evolutionary economics of mental time travel?. Trends
in cognitive sciences, 12(6), 219-224.
\item Bozeman, B., \& Rogers, J. D. (2002). A churn model of scientific
knowledge value: Internet researchers as a knowledge value collective.
Research Policy, 31(5), 769-794.
\item Brealey, R. A., Myers, S. C., Allen, F., \& Mohanty, P. (2012). Principles
of corporate finance. Tata McGraw-Hill Education.
\item Brown, D. J., \& Werner, J. (1995). Arbitrage and existence of equilibrium
in infinite asset markets. The Review of Economic Studies, 62(1),
101-114.
\item Buss, D. M. (1985). Human mate selection: Opposites are sometimes
said to attract, but in fact we are likely to marry someone who is
similar to us in almost every variable. American scientist, 73(1),
47-51.
\item Calaprice, A. (2000). The expanded quotable Einstein. Princeton, NJ:
Princeton.
\item Carmona, C. I. (2015). Jewelry Appraisal Handbook. The Journal of
Gemmology, 34(7), 639-640.
\item Carr, B. (Ed.). (2007). Universe or multiverse?. Cambridge University
Press.
\item Chiarella, C., \& He, X. (2001). Asset price and wealth dynamics under
heterogeneous expectations. Quantitative Finance, 1(5), 509-526.
\item Church, G. M., Gao, Y., \& Kosuri, S. (2012). Next-generation digital
information storage in DNA. Science, 1226355.
\item Cochrane, J. H. (2009). Asset Pricing:(Revised Edition). Princeton
university press.
\item Costanigro, M., McCluskey, J. J., \& Mittelhammer, R. C. (2007). Segmenting
the wine market based on price: hedonic regression when different
prices mean different products. Journal of agricultural Economics,
58(3), 454-466.
\item Courant, R., Robbins, H., \& Stewart, I. (1996). What is Mathematics?:
an elementary approach to ideas and methods. Oxford University Press,
USA.
\item Courtney, E. (2013). A Commentary on the Satires of Juvenal (No. 2).
Lulu. com.
\item Cramér, H. (2016). Mathematical methods of statistics (PMS-9) (Vol.
9). Princeton university press.
\item Creswell, J. W. (2002). Educational research: Planning, conducting,
and evaluating quantitative (pp. 146-166). Upper Saddle River, NJ:
Prentice Hall.
\item Crevoisier, O. (2016). The economic value of knowledge: Embodied in
goods or embedded in cultures?. Regional Studies, 50(2), 189-201.
\item Crouhy-Veyrac, L., Crouhy, M., \& Melitz, J. (1982). More about the
law of one price. European Economic Review, 18(2), 325-344.
\item Cruz Rambaud, S., \& Ventre, V. (2017). Deforming time in a nonadditive
discount function. International Journal of Intelligent Systems, 32(5),
467-480.
\item Dancy, J., Sosa, E., \& Steup, M. (Eds.). (2009). A companion to epistemology.
John Wiley \& Sons.
\item De Groot, R. S., Wilson, M. A., \& Boumans, R. M. (2002). A typology
for the classification, description and valuation of ecosystem functions,
goods and services. Ecological economics, 41(3), 393-408.
\item Delaney, C. J., Rich, S. P., \& Rose, J. T. (2016). A Paradox Within
The Time Value Of Money: A Critical Thinking Exercise For Finance
Students. American Journal of Business Education (Online), 9(2), 83.
\item Delanty, G. (2001). The university in the knowledge society. Organization,
8(2), 149-153.
\item De Beaumont, M. L. P. (1804). Beauty and the Beast. Prabhat Prakashan.
\item Derman, E. (2011). Models behaving badly. Why Conusion Illusion with
Reality Can Lead to Disaster, on Wall Street and in Life. J. Willey
and Sons.
\item DeRose, K. (2005). What is epistemology. A brief introduction to the
topic, 20.
\item Deutsch, D., \& Lockwood, M. (1994). The quantum physics of time travel.
Scientific American, 270(3), 68-74.
\item Diallo, S. Y., Padilla, J. J., Bozkurt, I., \& Tolk, A. (2013). Modeling
and simulation as a theory building paradigm. In Ontology, Epistemology,
and Teleology for Modeling and Simulation (pp. 193-206). Springer,
Berlin, Heidelberg.
\item Dixon, H. (1990). Equilibrium and explanation. Published in Creedy
J. The Foundations of Economic Thought. Blackwells. pp. 356–394.
\item Donovan, A., \& Laudan, R. (Eds.). (2012). Scrutinizing science: Empirical
studies of scientific change (Vol. 193). Springer Science \& Business
Media.
\item Einstein, A. (1956). Relativity: the special and the general theory
(pp. 106-106). Crown Publishers.
\item Fara, P. (1999). Catch a falling apple: Isaac Newton and myths of
genius. Endeavour, 23(4), 167-170.
\item Figueroa, A. (2016). Science Is Epistemology. In Rules for Scientific
Research in Economics (pp. 1-14). Palgrave Macmillan, Cham.
\item Foray, D., \& Lundvall, B. Ä. (1998). The knowledge-based economy:
from the economics of knowledge to the learning economy. The economic
impact of knowledge, 115-121. Chicago 
\item Foucault, M. (1984). What is Enlightenment?. The Foucault Reader,
32-50.
\item Frankham, R. (1995). Inbreeding and extinction: a threshold effect.
Conservation biology, 9(4), 792-799.
\item Frankham, R. (1997). Do island populations have less genetic variation
than mainland populations?. Heredity, 78(3), 311.
\item Frederick, S., Loewenstein, G., \& O'donoghue, T. (2002). Time discounting
and time preference: A critical review. Journal of economic literature,
40(2), 351-401.
\item Froot, K. A., Kim, M., \& Rogoff, K. (1995). The law of one price
over 700 years (No. w5132). National Bureau of Economic Research.
\item Fuchs, S. (1993). A sociological theory of scientific change. Social
forces, 71(4), 933-953.
\item Gerring, J. (2011). Social science methodology: A unified framework.
Cambridge University Press.
\item Goodwin, B. K., Grennes, T., \& Wohlgenant, M. K. (1990). Testing
the law of one price when trade takes time. Journal of International
Money and Finance, 9(1), 21-40.
\item Gott, J. R. (2002). Time travel in Einstein's universe: the physical
possibilities of travel through time. Houghton Mifflin Harcourt.
\item Gustafson, C. R., Lybbert, T. J., \& Sumner, D. A. (2016). Consumer
sorting and hedonic valuation of wine attributes: exploiting data
from a field experiment. Agricultural economics, 47(1), 91-103.
\item Halliday, D., \& Resnick, R. (1967). Physics Part I and II; John Wiley\&
Sons.
\item Hamilton, J. D. (1994). Time series analysis (Vol. 2, pp. 690-696).
Princeton, NJ: Princeton university press.
\item Harré, R. (1985). The philosophies of science.
\item Hassani, S. (2009). Dirac delta function. In Mathematical methods
(pp. 139-170). Springer, New York, NY.
\item Hayden, B. Y. (2016). Time discounting and time preference in animals:
a critical review. Psychonomic bulletin \& review, 23(1), 39-53.
\item Hetherington, S. C. (2018). Knowledge puzzles: An introduction to
epistemology. Routledge.
\item Highet, G. (1961). Juvenal the satirist: a study (Vol. 48). Oxford
University Press.
\item Isard, P. (1977). How far can we push the\textquotedbl{} law of one
price\textquotedbl ?. The American Economic Review, 67(5), 942-948.
\item Ismail, S. (2014). Exponential Organizations: Why new organizations
are ten times better, faster, and cheaper than yours (and what to
do about it). Diversion Books.
\item Iyer, B. R., \& Bhawal, B. (Eds.). (2013). Black Holes, Gravitational
Radiation and the Universe: Essays in Honor of CV Vishveshwara (Vol.
100). Springer Science \& Business Media.
\item Jehle, G. A., \& Reny, P. J. (2011). Advanced Microeconomic Theory.
Harlow, England, New York: Financial Times.
\item Kaku, M. (2009). Physics of the impossible: A scientific exploration
into the world of phasers, force fields, teleportation, and time travel.
Anchor.
\item Kant, I. (2013). Originally Written in 1784. An answer to the question:'What
is enlightenment?'. Penguin UK.
\item Kapferer, B. (2007). Anthropology and the Dialectic of Enlightenment:
A Discourse on the Definition and Ideals of a Threatened Discipline.
The Australian Journal of Anthropology, 18(1), 72-94.
\item Kashyap, R. (2015). A Tale of Two Consequences. The Journal of Trading,
10(4), 51-95.
\item Kashyap, R. (2017). Fighting Uncertainty with Uncertainty: A Baby
Step. Theoretical Economics Letters, 7(5), 1431-1452.
\item Kashyap, R. (2018). Nature vs Nurture and Science vs Art of Publishing,
Life and Everything Else. Working Paper, Social Science Research Network
(SSRN) Link: \href{https://papers.ssrn.com/sol3/papers.cfm?abstract_id=3172817}{Nature vs Nurture and Science vs Art of Publishing, Life and Everything Else}.
\item Kashyap, R. (2019). The perfect marriage and much more: Combining
dimension reduction, distance measures and covariance. Physica A:
Statistical Mechanics and its Applications, 536, 120938.
\item Kashyap, R. (2019b). For Whom the Bell (Curve) Tolls: A to F, Trade
Your Grade Based on the Net Present Value of Friendships with Financial
Incentives. The Journal of Private Equity, 22(3), 64-81.
\item Katz, J. S., \& Martin, B. R. (1997). What is research collaboration?.
Research policy, 26(1), 1-18.
\item Keller, L. F., \& Waller, D. M. (2002). Inbreeding effects in wild
populations. Trends in ecology \& evolution, 17(5), 230-241.
\item Kosuri, S., \& Church, G. M. (2014). Large-scale de novo DNA synthesis:
technologies and applications. Nature methods, 11(5), 499.
\item Kvanvig, J. L. (2003). The value of knowledge and the pursuit of understanding.
Cambridge University Press.
\item Laudan, L., Donovan, A., Laudan, R., Barker, P., Brown, H., Leplin,
J., ... \& Wykstra, S. (1986). Scientific change: Philosophical models
and historical research. Synthese, 69(2), 141-223.
\item Laudan, R., Laudan, L., \& Donovan, A. (1988). Testing theories of
scientific change. In Scrutinizing science (pp. 3-44). Springer, Dordrecht.
\item Lebreton, M., Jorge, S., Michel, V., Thirion, B., \& Pessiglione,
M. (2009). An automatic valuation system in the human brain: evidence
from functional neuroimaging. Neuron, 64(3), 431-439.
\item Leike, A. (2001). Demonstration of the exponential decay law using
beer froth. European Journal of Physics, 23(1), 21.
\item Levy, M., \& Levy, H. (1996). The danger of assuming homogeneous expectations.
Financial Analysts Journal, 52(3), 65-70.
\item Lewis, D. (1976). The paradoxes of time travel. American Philosophical
Quarterly, 13(2), 145-152.
\item Little, D. (1991). Varieties of social explanation: An introduction
to the philosophy of social science.
\item Loewenstein, G., \& Prelec, D. (1991). Negative time preference. The
American Economic Review, 81(2), 347-352.
\item Lutz, J. F., Ouchi, M., Liu, D. R., \& Sawamoto, M. (2013). Sequence-controlled
polymers. Science, 341(6146), 1238149.
\item MacMillan, D. M. (2008). Diwali: Hindu festival of lights. Enslow
Publishers, Inc..
\item Manicas, P. T. (1991). History and philosophy of social science.
\item Mantzicopoulos, P., \& Patrick, H. (2010). “The seesaw is a machine
that goes up and down”: Young children's narrative responses to science-related
informational text. Early Education and Development, 21(3), 412-444.
\item Martin, J. (1996). Cybercorp: the new business revolution. American
Management Assoc., Inc.
\item Marty, W. (2017). The Time Value of Money. In Fixed Income Analytics
(pp. 5-16). Springer, Cham.
\item Marzilli Ericson, K. M., White, J. M., Laibson, D., \& Cohen, J. D.
(2015). Money earlier or later? Simple heuristics explain intertemporal
choices better than delay discounting does. Psychological science,
26(6), 826-833.
\item Masset, P., \& Weisskopf, J. P. (2015). Wine funds: an alternative
turning sour?. The Journal of Alternative Investments, 17(4), 6-20.
\item McKie, D., \& De Beer, G. R. (1951). Newton's apple. Notes and Records
of the Royal society of London, 9(1), 46-54.
\item Meadows, D. H., Meadows, D. L., Randers, J., \& Behrens, W. W. (1972).
The limits to growth. New York, 102, 27.
\item Melin, G. (2000). Pragmatism and self-organization: Research collaboration
on the individual level. Research policy, 29(1), 31-40.
\item Miotti, L., \& Sachwald, F. (2003). Co-operative R\&D: why and with
whom?: An integrated framework of analysis. Research policy, 32(8),
1481-1499.
\item Morgan, M. S., \& Morrison, M. (Eds.). (1999). Models as mediators:
Perspectives on natural and social science (Vol. 52). Cambridge University
Press.
\item Nagel, T. (2012). Mortal questions. Cambridge University Press.
\item Nahin, P. J. (2001). Time machines: Time travel in physics, metaphysics,
and science fiction. Springer Science \& Business Media.
\item Nahin, P. J. (2017). The Physics of Time Travel: II. In Time Machine
Tales (pp. 289-337). Springer, Cham.
\item Nelson, R. R. (1959). The simple economics of basic scientific research.
Journal of political economy, 67(3), 297-306.
\item Newman, J. (2017). Diwali (Festivals in Different Cultures). The School
Librarian, 65(2), 89.
\item Nikodým, O. M. (1966). Dirac’s Delta-function. In The Mathematical
Apparatus for Quantum-Theories (pp. 724-760). Springer, Berlin, Heidelberg.
\item Norris, C. (1994). What is enlightenment? Kant according to Foucault.
The Cambridge Companion to Foucault, 159-196.
\item O'hear, A. (1993). An introduction to the philosophy of science.
\item Olson, M., \& Bailey, M. J. (1981). Positive time preference. Journal
of Political Economy, 89(1), 1-25.
\item Ortiz, C. E., Stone, C. A., \& Zissu, A. (2017). Beer Annuities: Hold
the Interest and Principal. The Journal of Derivatives, 24(4), 108-114.
\item Papineau, D. (2002). Philosophy of science. The Blackwell companion
to philosophy, 286-316.
\item Pavitt, K. (1991). What makes basic research economically useful?.
Research policy, 20(2), 109-119.
\item Petters, A. O., \& Dong, X. (2016). The Time Value of Money. In An
Introduction to Mathematical Finance with Applications (pp. 13-82).
Springer, New York, NY.
\item Pitkethly, R. (1997). The valuation of patents: a review of patent
valuation methods with consideration of option based methods and the
potential for further research. Research Papers in Management Studies-University
of Cambridge Judge Institute of Management Studies.
\item Plomin, R., \& Daniels, D. (1987). Why are children in the same family
so different from one another?. Behavioral and brain Sciences, 10(1),
1-16.
\item Pritchard, D. (2009). The value of knowledge. The Harvard Review of
Philosophy, 16(1), 86-103.
\item Pritchard, D., Millar, A., \& Haddock, A. (2010). The nature and value
of knowledge: Three investigations. OUP Oxford.
\item Pritchard, D. (2018). What is this thing called knowledge?. Routledge.
\item Protopapadakis, A., \& Stoll, H. R. (1983). Spot and futures prices
and the law of one price. The Journal of Finance, 38(5), 1431-1455.
\item Psillos, S. (2005). Scientific realism: How science tracks truth.
Routledge.
\item Randall, L. (2006). Warped Passages: Unravelling the universe's hidden
dimensions. Penguin UK.
\item Reagan, M. D. (1967). Basic and applied research: a meaningful distinction?.
Science, 155(3768), 1383-1386.
\item Rescigno, A., Beck, J. S., \& Thakur, A. K. (1987). The use and abuse
of models. Journal of pharmacokinetics and biopharmaceutics, 15(3),
327-340.
\item Rao, C. R. (1973). Linear statistical inference and its applications.
New York: Wiley.
\item Rosenberg, A. (2018). Philosophy of social science. Routledge.
\item Ross, S. A., Westerfield, R. W., \& Jaffe, J. F. (2002). Corporate
Finance.
\item Rosser, M., \& Lis, P. (2016). Basic mathematics for economists. Routledge.
\item Roy, R. K., Meszynska, A., Laure, C., Charles, L., Verchin, C., \&
Lutz, J. F. (2015). Design and synthesis of digitally encoded polymers
that can be decoded and erased. Nature communications, 6, 7237.
\item Russell, B. (1948). Human knowledge: its scope and value. Routledge.
Chicago.
\item Sagan, C. (2006). Cosmos. Edicions Universitat Barcelona.
\item Sandel, M. J. (2012). What money can't buy: the moral limits of markets.
Macmillan.
\item Sankey, H. (2016). Scientific realism and the rationality of science.
Routledge.
\item Shankar, R. (2012). Principles of quantum mechanics. Springer Science
\& Business Media.
\item Shapere, D. (1980). The character of scientific change. In Scientific
discovery, logic, and rationality (pp. 61-116). Springer, Dordrecht.
\item Shapere, D. (1989). Evolution and continuity in scientific change.
Philosophy of science, 56(3), 419-437.
\item Shepard, H. A. (1956). Basic research and the social system of pure
science. Philosophy of Science, 23(1), 48-57.
\item Shleifer, A., \& Vishny, R. W. (1997). The limits of arbitrage. The
Journal of Finance, 52(1), 35-55.
\item Slatkin, M. (1987). Gene flow and the geographic structure of natural
populations. Science, 236(4803), 787-792.
\item Smart, J. J. C. (2014). Philosophy and scientific realism. Routledge.
\item Stocker, T. F. (1998). The seesaw effect. Science, 282(5386), 61-62.
\item Stukeley, W. (1936). Memoirs of Sir Isaac Newton's Life. Taylor and
Francis.
\item Suddendorf, T., \& Busby, J. (2005). Making decisions with the future
in mind: Developmental and comparative identification of mental time
travel. Learning and Motivation, 36(2), 110-125.
\item Suddendorf, T., \& Busby, J. (2003). Mental time travel in animals?.
Trends in cognitive sciences, 7(9), 391-396.
\item Suddendorf, T., \& Corballis, M. C. (2007). The evolution of foresight:
What is mental time travel, and is it unique to humans?. Behavioral
and brain sciences, 30(3), 299-313.
\item Suddendorf, T., Addis, D. R., \& Corballis, M. C. (2011). Mental time
travel and shaping of the human mind. M. Bar, 344-354.
\item Thaler, R. (1981). Some empirical evidence on dynamic inconsistency.
Economics letters, 8(3), 201-207.
\item Thorne, K. (2014). The science of Interstellar. WW Norton \& Company.
\item Toffler, A. (1990). Powershift. New York: Bantam.
\item Throsby, D. (2003). Determining the value of cultural goods: How much
(or how little) does contingent valuation tell us?. Journal of cultural
economics, 27(3-4), 275-285.
\item Tweney, R. D., Doherty, M. E., \& Mynatt, C. R. (1981). On scientific
thinking.
\item Wagenaar, W. A., \& Sagaria, S. D. (1975). Misperception of exponential
growth. Perception \& Psychophysics, 18(6), 416-422.
\item Wu, M. C., \& Tseng, C. Y. (2006). Valuation of patent–a real options
perspective. Applied Economics Letters, 13(5), 313-318.
\item Valsiner, J. (2007). Culture in minds and societies: Foundations of
cultural psychology. Psychol. Stud.(September 2009), 54, 238-239.
\item Varian, H. R. (1992). Microeconomic analysis.
\item Weinberg, S. (2007). Living in the multiverse. Universe or multiverse,
29-42.
\item Woodward, J. F. (1995). Making the universe safe for historians: Time
travel and the laws of physics. Foundations of Physics Letters, 8(1),
1-39.
\item Zhao, W., \& Zhou, X. (2011). Status inconsistency and product valuation
in the California wine market. Organization Science, 22(6), 1435-1448.
\item Zucker, L. G., Darby, M. R., \& Armstrong, J. S. (2002). Commercializing
knowledge: University science, knowledge capture, and firm performance
in biotechnology. Management science, 48(1), 138-153.\pagebreak{}
\end{doublespace}
\end{enumerate}
\begin{doublespace}

\section{Figures}
\end{doublespace}

\begin{doublespace}
\begin{figure}
\includegraphics[width=9cm]{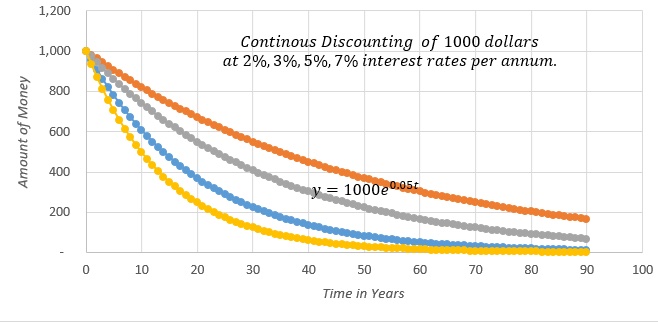}\includegraphics[width=9cm]{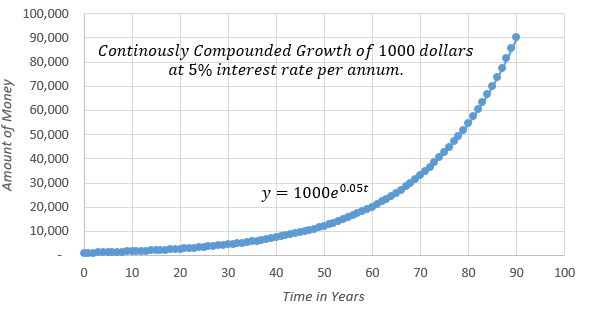}

\caption{\label{fig:Time-Value-of}Time Value of Money: Discounting and Compounding}
\end{figure}

\begin{figure}
\includegraphics[width=9cm]{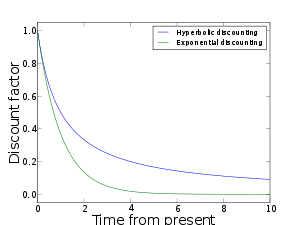}\includegraphics[width=9cm]{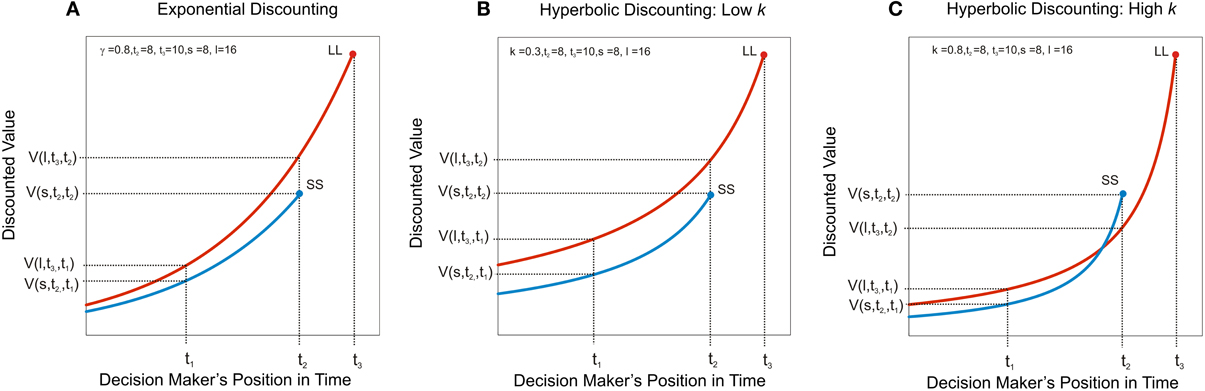}

\caption{\label{fig:Discount-Functions:-Hyperbolic}Discount Functions: Hyperbolic
and Exponential}
\end{figure}

\begin{figure}
\includegraphics[width=6.5cm]{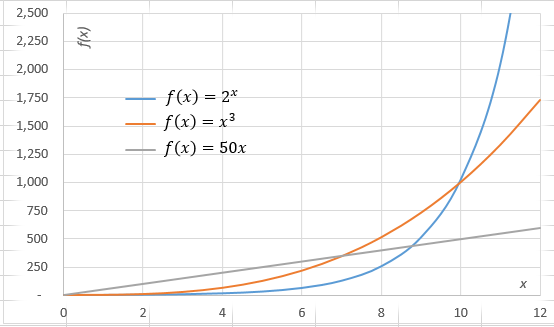}\includegraphics[width=5cm]{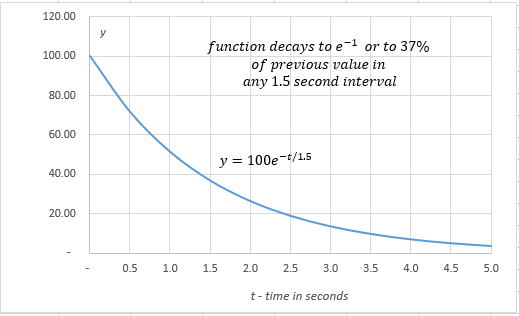}\includegraphics[width=6.5cm]{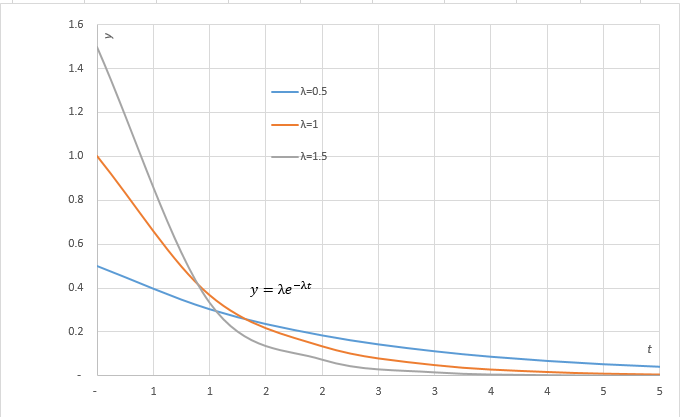}

\caption{\label{fig:Exponential-Growth-Decay}Exponential Growth and Decay}
\end{figure}

\begin{figure}
\includegraphics[width=9cm]{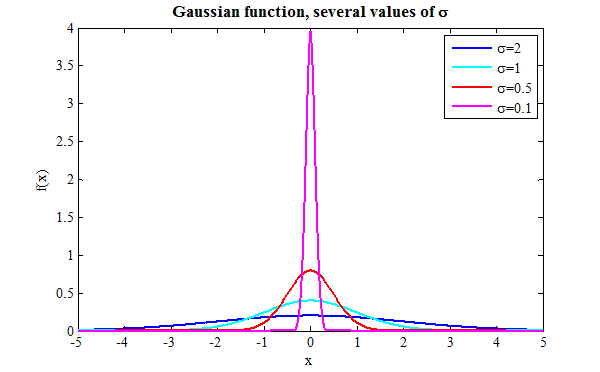}\includegraphics[width=9cm]{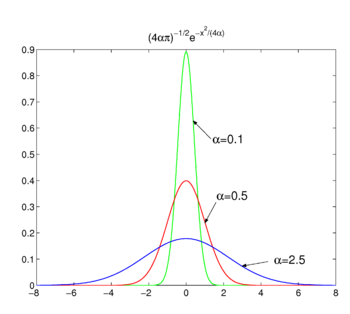}

\caption{\label{fig:Gaussian-Weight-Functions}Gaussian Weight Functions}
\end{figure}
\begin{figure}
\includegraphics[width=6cm]{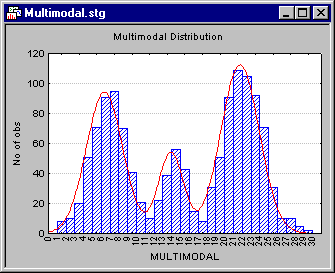}\includegraphics[width=6cm]{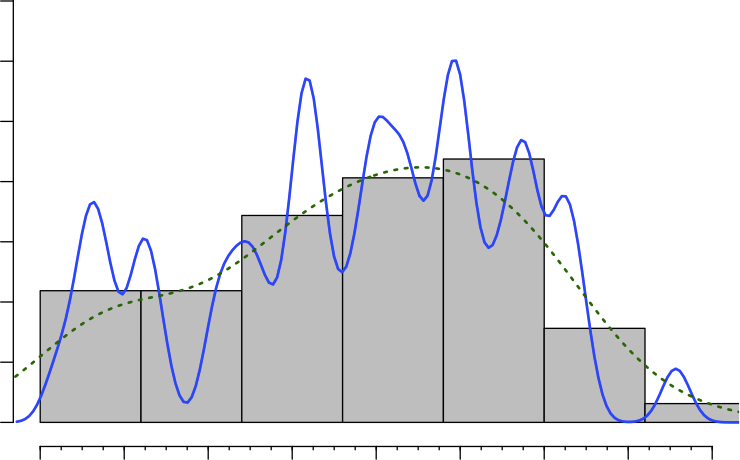}\includegraphics[width=6cm]{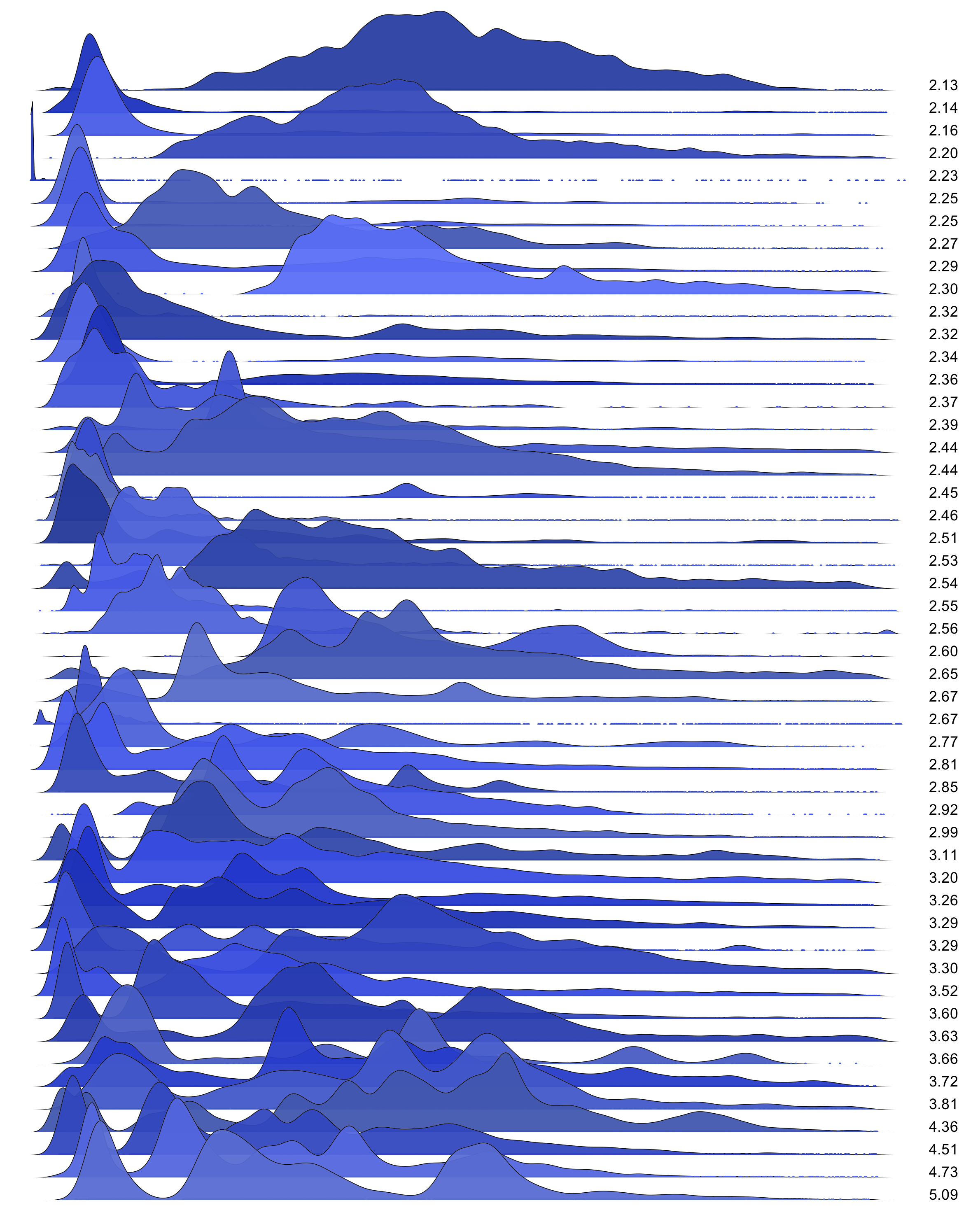}

\caption{\label{fig:Multi-Modal-Weight-Functions}Multi-Modal Weight Functions}
\end{figure}
\begin{figure}
\includegraphics[width=6cm]{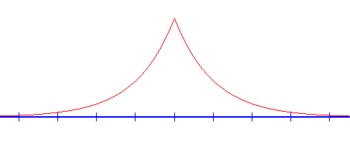}\includegraphics[width=12cm]{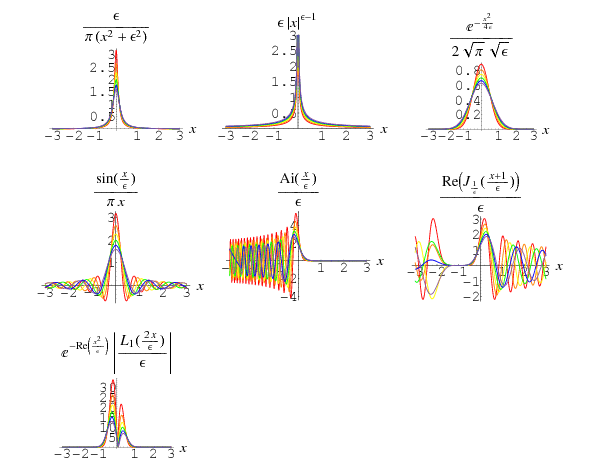}

\caption{\label{fig:Impulse-Functions}Impulse Functions}
\end{figure}
\begin{figure}
\includegraphics[width=18cm]{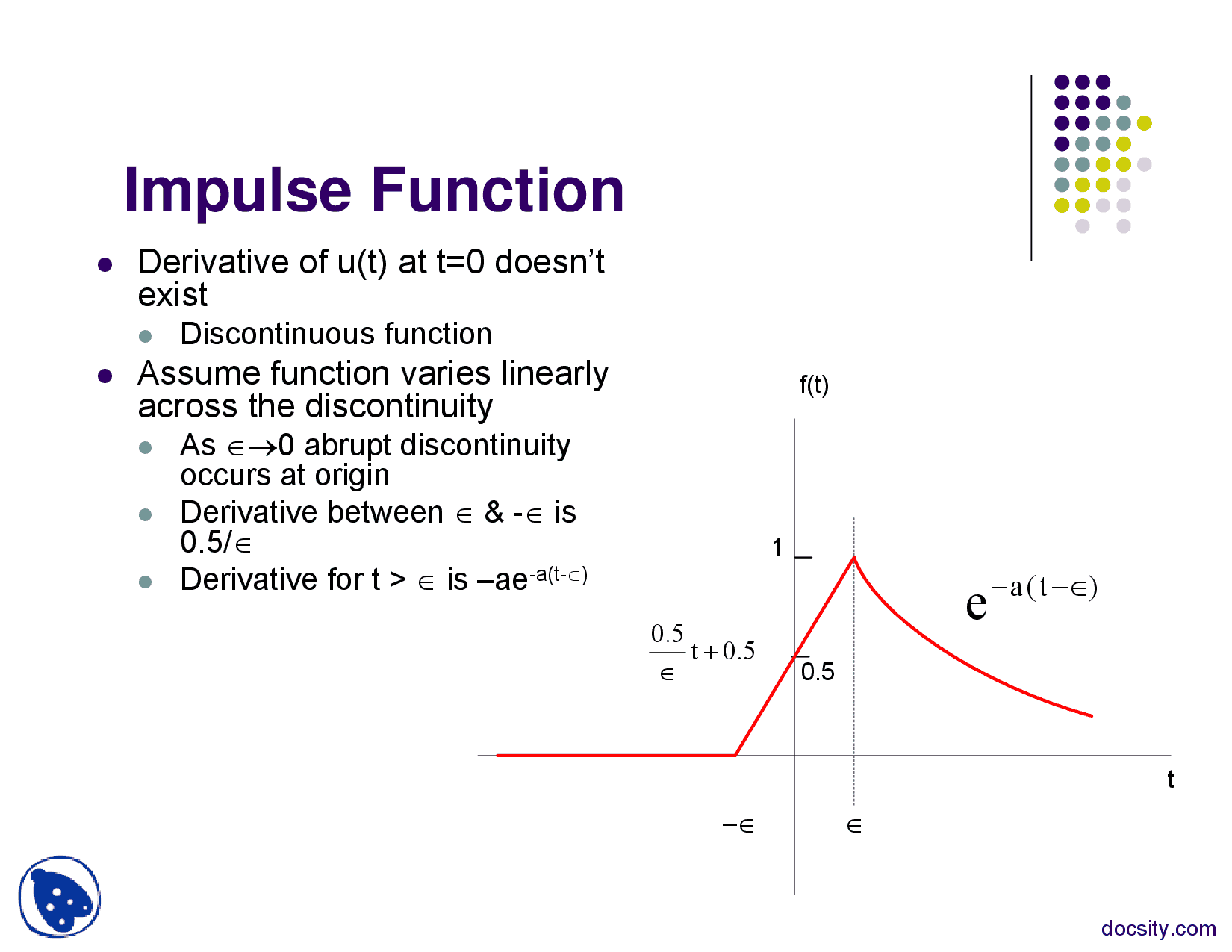}

\caption{\label{fig:Impulse-Function-with-Linear}Impulse Function with Linear
Discontinuity}
\end{figure}
\end{doublespace}

\end{document}